\documentclass[11pt,a4paper]{article}
\usepackage{amsmath,amssymb,amsthm}
\usepackage{geometry}
\usepackage{hyperref}
\usepackage{graphicx}
\usepackage{braket}
\usepackage{authblk}
\usepackage{enumitem}
\usepackage{booktabs}
\usepackage{tabularx}
\usepackage{float}
\usepackage{mathrsfs}
\usepackage[utf8]{inputenc}
\usepackage[T1]{fontenc}
\usepackage{comment}
\usepackage[backend=biber,style=numeric, maxnames=99,
  minnames=99]{biblatex}
\newtheorem{theorem}{Theorem}[section]
\newtheorem{corollary}[theorem]{Corollary}
\newtheorem{lemma}[theorem]{Lemma}
\newtheorem{definition}[theorem]{Definition}
\newtheorem{example}[theorem]{Example}

\newtheorem{remark}[theorem]{Remark}
\newtheorem*{unnumberedremark}{Remark}

\newtheorem{proposition}[theorem]{Proposition}

\newcommand{\g}{\mathfrak{g}}

\newcommand{\s}{\mathfrak{s}}
\newcommand{\rfrak}{\mathfrak{r}}
\newcommand{\zfrak}{\mathfrak{z}}
\newcommand{\K}{\mathcal{K}}
\newcommand{\su}{\mathfrak{su}}
\newcommand{\uu}{\mathfrak{u}}

\newcommand{\Tr}{\mathrm{Tr}}
\newcommand{\ad}{\mathrm{ad}}
\newcommand{\ImC}{\mathrm{Im}(C)}
\newcommand{\Rad}{\mathrm{Rad}}

\newcommand{\C}{\mathbb{C}}

\title{A Quantum Algorithm for the Radical of a Lie Algebra:\\
   Kernel Projection and Conditioning}
\author{Yibin Wang}
\affil[1]{Graduate School of Mathematics, Nagoya University, Nagoya, 464-8601, Aichi, Japan.\\
\texttt{yibinw0210@gmail.com}}
\date{}

\begin{document}

\maketitle

\begin{abstract}
Every finite-dimensional real or complex Lie algebra has a largest solvable ideal,
its radical. In the compact dynamical Lie algebra (DLA) of a closed quantum system,
this radical is the center; projecting onto it isolates directions that commute
with the supplied algebra. With sparse Lie-bracket data and a known spectral gap,
we construct a quantum circuit that approximately encodes this coefficient-space
projector. The general construction combines the derived-algebra map with the
Killing form, whereas compactness makes the Killing form alone sufficient. This
distinction matters for numerical sensitivity, and our main theorem
gives its exact law under an invariant-orthonormal basis adapted to the center and
semisimple part. In a compact real algebra with a nonzero semisimple part, the
condition number of the general operator is the three-halves power of the Killing-form
condition number on that part. The compact projector in turn yields a bounded-error
test for whether the center is trivial, provided that any nonzero center has at
least a stated minimum dimension. This gives a controlled test for internal
conserved directions in compact quantum dynamics.
\end{abstract}

\newpage
\tableofcontents

\section{Introduction}

Lie groups and Lie algebras give the language for the continuous symmetries and the dynamical evolution of quantum systems~\cite{Hall2015, georgi2000lie}. The Levi--Malcev theorem~\cite{humphreys2012introduction, jacobson2013lie} lies at the base of that theory. It states that any finite-dimensional Lie algebra $\g$ over a field of characteristic zero splits as a semidirect product $\g = \s \ltimes \rfrak$ into a semisimple Levi factor $\s$ and the unique maximal solvable ideal $\rfrak$, the radical. Identifying these two components is the entry point to any structural or dynamical analysis of $\g$.

The decomposition presents two computational problems: identifying the unique radical $\rfrak$ and constructing a Levi factor $\s$. This paper develops a quantum method for the first, formulated as radical-projector block encoding and then specialized to compact quantum dynamics.

\subsection{Summary of Contributions}

The paper develops a quantum method for the solvable radical and studies the
compact case. Three results follow.

\begin{enumerate}
    \item \textbf{Quantum access to the radical.} From Lie-bracket data, we construct a quantum circuit that encodes the projector onto the solvable radical of a real or complex Lie algebra. This avoids supplying the dense Killing-form matrix as input (Theorem~\ref{thm:general_complexity_main}). For compact Lie algebras, the radical is the center and the circuit acts directly on the Killing form (Corollary~\ref{cor:compact_optimization}). Section~\ref{sec:optimized_classical_benchmark} compares it with a classical procedure using the same structural information.

    \item \textbf{A test for a trivial center.} Using this projector, we test whether the center of a compact dynamical Lie algebra is trivial, under the promise that any nontrivial center has at least a stated minimum dimension (Theorem~\ref{thm:zero_center_decision}). The leading cost grows with the square root of the algebra dimension relative to that minimum. An explicit oracle family (Theorem~\ref{thm:opaque_value_separation}) connects this behavior to quantum search~\cite{brassard2000quantum,NayakWu1999}. The output is only a yes-or-no answer; finding a center basis or its exact dimension are separate tasks.

    \item \textbf{Main theorem: an exact conditioning law.} Theorem~\ref{thm:obstruction_general} is the paper's principal structural result. Under its stated basis assumptions, it gives an exact law for numerical sensitivity in which the condition number of the general radical-identification operator is the three-halves power of the corresponding Killing-form condition number on the semisimple part. The compact construction remains distinct because it acts directly on the Killing form.

\end{enumerate}

The cited results and the input/output paragraph below contain the precise
assumptions and guarantees.

\subsection{The Central Role of the Radical in Quantum Dynamics}\label{sec:radical_role_quantum_physics}

The radical takes on direct physical content here. The Dynamical Lie Algebra (DLA) of a closed quantum system inherits constraints from the unitarity of time evolution, and under those constraints the general Levi-Malcev decomposition collapses to a simpler form in which the radical becomes a concrete physical object~\cite{d2021introduction, d2008lie}.

In quantum mechanics, we analyze the DLA $\g = \text{Lie}(\{iH_k\})$ generated by a set of Hamiltonians $\{H_k\}$, for example, control Hamiltonians. Since these generators are skew-Hermitian matrices ($iH_k$), the DLA is a subalgebra of the unitary algebra $\uu(N)$, where $N=2^n$ for $n$ qubits. The unitary group $U(N)$ is compact. Consequently, the DLA of any finite-dimensional closed quantum system is a compact Lie algebra. 
Compactness specializes the Levi--Malcev structure. Any compact Lie algebra $\g$ is a direct sum of its center $\zfrak$ and a semisimple subalgebra $\s = [\g,\g]$~\cite{humphreys2012introduction}:

\begin{equation}\label{eq:compact_dla_structure}
\g = \s \oplus \zfrak.
\end{equation}
The radical equals the Abelian center $\zfrak$, and $[\s,\zfrak]=0$. For any $H = H_s + H_z \in \g$, where $H_s \in \s$ and $H_z \in \zfrak$, the two components commute and their evolution factorizes exactly. The split therefore carries zero cross-term Trotter error between $\s$ and $\zfrak$; the general Levi--Malcev case, where $[\s,\rfrak]\ne 0$, instead requires a product formula~\cite{Trotter1959, Suzuki1976}. The central evolution $e^{-iH_z t}$ can still be physically nontrivial. Implementing this coefficient-space decomposition as a simulation circuit requires a separately specified coefficient-to-operator map with its own normalization and error analysis.

The center $\zfrak$ consists of elements that commute with the entire DLA and therefore describes conserved generators internal to it. The full commutant is a larger object: it can contain charges outside $\g$ and is the appropriate setting for a complete account of physical conservation laws or decoherence-free subsystems.

\subsection{Classical Algorithms and Their Limits}\label{sec:classical_limitations}

For an $n$-qubit system the DLA dimension can reach $d=O(4^{n})$~\cite{Nielsen2010}. Explicit dense representations can therefore become large even before the structural calculation begins.

Classical algorithms for radical computation use the structural identity $\rfrak = [\g,\g]^{\perp_{\K}}$ (Theorem~\ref{thm:radical_orthogonal_complement}) in the method developed by de~Graaf and collaborators~\cite{deGraaf2000, deGraaf1997Levi}. Section~\ref{sec:optimized_classical_benchmark} gives a self-contained operation count under the same explicit structure data. Constructing $K=A^{T}B$ costs $O(d\,s_{adj}s_{int}+d^{2})$ arithmetic operations, including dense-output writes. A general $d\times d^{2}$ rank calculation then gives the safe bound $O(d\,s_{adj}s_{int}+d^{4})$. Compactness replaces that rectangular step by $\ker K$ and gives $O(d\,s_{adj}s_{int}+d^{3})$. The corresponding bit bounds are stated separately. For an actual Pauli label set $S$, the center is the real span of $S\cap V^{\perp_{\omega}}$, where $V=\operatorname{span}_{\mathbb F_2}S$. Row reduction and label testing cost $O(dn^2)$ bit operations, or $O(dn)$ word operations with $\Theta(n)$-bit packed words. This becomes $\widetilde O(d)$ when $n=O(\operatorname{polylog}d)$.

\subsection{A Quantum Algorithm via Direct Projection}

Our algorithm implements the Cartan-criterion identity as a kernel projection on a quantum block encoding. The general construction composes $C^{T}$ with a factorized encoding of $K$, while the compact construction projects directly onto $\ker K$.

The starting point is the identity $\rfrak = ([\g, \g])^{\perp_{\K}}$ (Theorem~\ref{thm:radical_orthogonal_complement}). It converts radical finding into the kernel of an operator built from the structure constants (Section~\ref{sec:radical_as_kernel}). Quantum singular value transformation (QSVT)~\cite{Gilyen2019} implements a zero-singular-value filter once a separating gap is supplied. Factoring $K=A^{T}B$ avoids treating the dense Killing matrix as the primitive input (Section~\ref{sec:killing_form_factorization}).

For compact inputs, Section~\ref{sec:zero_center_decision} specializes
normalized projector-rank estimation to a zero-center promise test. A valid
maximally entangled coefficient state converts the projector success probability
into the normalized center dimension. This step follows earlier all-direction
kernel-detection and projector-rank methods
~\cite{ambainis2012variable,mcardle2026streamlined}; the result here binds that
primitive to the compact-DLA center projector and tracks the valid coefficient
domain and approximation-error budget required by the stated promise.

\subsection{Related Works and Assumptions}

Projector-rank estimation and all-direction kernel tests precede the decision
wrapper used here. Ambainis used a maximally mixed probe to test matrix singularity
under a spectral promise~\cite{ambainis2012variable}. McArdle, Gily\'en, and
Berta formulated normalized projector-rank estimation explicitly and used a
maximally entangled purification when the reference projector is the identity
~\cite{mcardle2026streamlined}.

Quantum center computation has also been studied in a different algebraic model.
Hern\'andez C\'aceres, R\'ua, and Combarro computed the centers of
finite-dimensional algebras over finite fields by Abelian hidden-subgroup methods
~\cite{hernandez2023substructures}. For a Lie product in odd characteristic,
their commutator condition agrees with the usual Lie-center condition, and a
central element satisfies their nucleus conditions automatically. Their result
concerns a full multiplication table over a finite field. The present setting is
a compact real DLA with sparse typed-value access, so the two works establish
complementary results in distinct input models.

Quantum query bounds for other multiplication-table properties, including
commutativity, and sparse Lie structure-constant access have also appeared
~\cite{combarro2019commutativity,somma2025shadow}. These works set the
property-testing and access-model context. Their tasks differ from the compact-DLA
zero-center promise studied here.

Earlier quantum algorithms for algebraic problems mostly worked on sharply structured inputs, the hidden subgroup problem being the canonical example~\cite{Lomont2004}. Quantum linear algebra methods, such as HHL~\cite{Harrow2009} and QSVT~\cite{Gilyen2019}, broadened that scope to matrix problems. The present construction uses sparse access to factor matrices derived from the structure tensor; the access model and its setup are part of the theorem statement.

\paragraph{Relation to adjacent research areas.}

Lie algebras enter quantum computing in several adjacent regimes. One line studies the DLA as a predictor of Variational Quantum Algorithm trainability: the dimension of the DLA controls the variance of cost gradients and the onset of barren plateaus~\cite{Larocca2022diagnosing, Ragone2024lie}, separately from the locality of the cost function~\cite{Cerezo2021}. A second line uses Lie-algebraic structure to speed up \emph{classical} simulation of quantum dynamics~\cite{Zeier2011}. The present task uses a quantum computer to extract the algebraic structure itself.

The closest classical line treats Lie-algebraic structure as the object of
computation. Algorithms that recover structural data from structure constants,
including Levi factors and nilradicals, go back to Rand, Winternitz and
Zassenhaus~\cite{RandWinternitzZassenhaus1988} and to the computer-algebra work of
de~Graaf and collaborators~\cite{deGraaf2000, deGraaf1997Levi}.

Quantum work addresses several neighboring tasks. The DLA of a fixed ansatz may
be derived analytically, as in the QAOA cycle-graph example of Allcock and
coworkers~\cite{Allcock2024QAOA}. Cartan (KAK) methods compile a given unitary into
fixed-depth circuits~\cite{Kokcu2022Cartan}, while conserved quantities of unknown
dynamics can be learned from measured data~\cite{ConservationLaws2024}. Here the
input is a supplied Lie-algebra basis with its structure data, and the output is a
block encoding of the radical projector in coefficient space. For the full Pauli algebra, the
symplectic form also supplies the required support and value access on demand.

Section~\ref{sec:simulability_significance} compares this structural task with recent classical-simulability results.

\paragraph{The DLA generation problem and input/output model.}

Physical systems are often specified by a small generating set, whereas the projector theorem takes a fixed full DLA basis and bracket rule as input. Generating or recognizing that algebra lies upstream of the stated projector query bound.

The numerical input is over $\mathbb F\in\{\mathbb R,\mathbb C\}$. It consists of a named basis, exact zero patterns, and $b$-bit fixed-point values for the structure constants. The classical metadata also supplies a certified $b$-bit magnitude bound $F\geq\max_{i,j,k}|f_{ijk}|$ and the padded support bounds $s_{adj},s_{int},s_{\mathrm{br}},s_{\mathrm{out}}$ used by the relevant enumerators. Reversible row-support, column-support, and value oracles are supplied for $B$, $A^{T}$, and, in the general construction, $C^{T}$. The inverse of each support oracle is available at the same query cost. Let $g_{\mathrm{acc}}$ be an upper bound on the reversible-gate cost of one such support or value query. A QRAM implementation may compute the metadata and build those oracles from the explicit tensor at $O(d^{3})$ setup cost. An implicit implementation must supply the same certified metadata and support-enumeration interface; an entry or sign oracle alone is insufficient. The full Pauli algebra supplies this interface with $F=2$ by reversible symplectic enumeration (Section~\ref{sec:arithmetic_oracles}). An arbitrary Pauli subset does so only if a reversible rank/select oracle for its actual support is additionally provided.

The promise gives a certified $b_{\Delta}$-bit fixed-point lower bound $\Delta>0$ such that every singular value of the chosen target lies in $\{0\}\cup[\Delta,\infty)$. The represented value of $\Delta$ is the threshold used by the circuit; its validity is an input promise. The tolerance $\epsilon\in(0,1/2)$ specifies operator-norm error. The primary output is a circuit that block-encodes the Euclidean coefficient-space projector $P_{\rfrak}$ within error $\epsilon$. Estimating $\Delta$, counting the kernel dimension, recovering a classical basis, preparing a coefficient state, and mapping coefficient vectors to physical operators are separate tasks. The cost unit is one call to a typed support or value oracle, with the additional reversible-gate cost stated separately. The success statement concerns deterministic circuit approximation. Writing $\widetilde P_{\rfrak}$ for the returned signal block, postselection after applying it has probability $p_{\psi}=\|\widetilde P_{\rfrak}\ket{\psi}\|^{2}$ for a normalized input, with $|p_{\psi}-\|P_{\rfrak}\ket{\psi}\|^{2}|\leq2\epsilon$. Section~\ref{sec:zero_center_decision} adds a decision wrapper under additional preparation and controlled-access assumptions.

\section{Notation and Background}\label{sec:lie_algebras_structure_constants}

We work with finite-dimensional Lie algebras over $\mathbb{R}$ or $\mathbb{C}$ and use the notation, structural parameters, and basic fact introduced below. Standard references are~\cite{humphreys2012introduction, jacobson2013lie, georgi2000lie}.

\subsection{Lie Algebras and Structure Constants} 
Throughout this paper, $\g$ denotes a finite-dimensional Lie algebra of dimension $d$ with Lie bracket $[\cdot,\cdot]: \g \times \g \rightarrow \g$. The Levi-Malcev decomposition $\g = \s \ltimes \rfrak$ separates $\g$ into a semisimple Levi part $\s$ and the solvable radical $\rfrak$~\cite{Maltsev1942}. Given a basis $\{B_i\}_{i=1}^{d}$ for $\g$, the structure constants $f_{ijk} \in \mathbb{F}$ are determined by

\begin{equation}\label{eq:structure_constants_definition}
[B_{i}, B_{j}] = \sum_{k=1}^{d} f_{ijk} B_{k}.
\end{equation}

These $d^{3}$ values are the fundamental input to the algorithm.

For any $x \in \g$ the linear map $\ad(x): \g \rightarrow \g$, $\ad(x)(y) := [x,y]$, has matrix elements in the basis $\{B_i\}$ given by
\begin{equation}\label{eq:adjoint_matrix_elements}
(\ad(B_{i}))_{kj} = f_{ijk}.
\end{equation}
The Killing form $\K(x,y) := \Tr(\ad(x)\ad(y))$ is bilinear, and its matrix elements expand to
\begin{equation}\label{eq:killing_form_matrix_elements}
K_{ij} = \sum_{k=1}^{d}\sum_{m=1}^{d} f_{ikm}\,f_{jmk}.
\end{equation}
This double-sum form is the target of the bilinear factorization $K = A^T B$ developed in Section~\ref{sec:killing_form_factorization} and exploited by both the quantum algorithm and the structure-aware classical benchmark.

\begin{definition}[Standard Pauli Basis (Skew-Hermitian)]\label{def:pauli_basis_convention}
For an $n$-qubit system ($N=2^n$), the standard Pauli operators $\{P_k\}$ are the $N \times N$ Hermitian matrices formed by tensor products of the single-qubit Pauli matrices $\{I, X, Y, Z\}$. We define the Skew-Hermitian Pauli Basis $\{B_k\}$ for the unitary algebra $\uu(N)$ by multiplying the standard Pauli operators by the imaginary unit $i$: $B_k = i P_k$.
\end{definition}

This convention keeps the basis elements skew-Hermitian (consistent with the generators of unitary evolution).

\begin{lemma}[Reality and Integrality of Structure Constants]\label{lemma:pauli_constants_real}
In the Skew-Hermitian Pauli Basis, the structure constants $f_{ijk}$ are real integers belonging to $\{0, \pm 2\}$.
\end{lemma}
\begin{proof}
The commutator of two standard Pauli operators $P_i, P_j$ satisfies $[P_i, P_j] = 2i \sum_k \epsilon_{ijk} P_k$ if they anti-commute and form a triplet (e.g., $X, Y, Z$), or $[P_i, P_j]=0$ if they commute. The tensor $\epsilon_{ijk}$ takes values in $\{0, \pm 1\}$.

We examine the commutator in the Skew-Hermitian basis:
\[
[B_i, B_j] = [iP_i, iP_j] = -[P_i, P_j] = - \left(2i \sum_k \epsilon_{ijk} P_k\right).
\]
Rearranging the terms:
\[
[B_i, B_j] = -2 \sum_k \epsilon_{ijk} (i P_k) = -2 \sum_k \epsilon_{ijk} B_k.
\]
The structure constants are therefore $f_{ijk} = -2\epsilon_{ijk}$. Since $\epsilon_{ijk} \in \{0, \pm 1\}$, the structure constants $f_{ijk}$ are real integers in $\{0, \pm 2\}$.
\end{proof}

\subsection{Complexity Parameters and Basis Properties} 
Both the quantum runtime and the structure-aware classical benchmark are controlled by four parameters drawn from the $f_{ijk}$ tensor.

\begin{definition}[Certified Structure Constant Magnitude Bound ($F$)]\label{def:max_f}
The input magnitude $F$ is a supplied $b$-bit fixed-point number satisfying
\begin{equation}
F \geq \|f\|_{\max} = \max_{i,j,k} |f_{ijk}|.
\end{equation}
Its validity is a promise. When the exact maximum is known, it may be used as $F$.
\end{definition}

For the Pauli basis, Lemma~\ref{lemma:pauli_constants_real} establishes $F=2$.

\begin{definition}[Standardized Basis]\label{def:standardized_basis}
A basis is \emph{standardized} when it admits a certified $F=O(1)$.
\end{definition}

We now define parameters that characterize the sparsity of the algebraic interactions. These parameters control both the quantum block encoding normalization factors and the optimized classical Sparse Matrix Multiplication (SpMM).

\begin{definition}[Adjoint Sparsity ($s_{adj}$)]\label{def:adjoint_sparsity}
The padded adjoint sparsity $s_{adj}$ is a supplied bound on the number of non-zero structure constants associated with any single basis element acting as the first argument of the commutator:
\begin{equation}\label{eq:adjoint_sparsity}
s_{adj} \geq \max_{i} |\{(j,k) : f_{ijk} \neq 0\}|.
\end{equation}
\end{definition}

This parameter corresponds to the maximum sparsity of the adjoint representation matrices, see Equation~\eqref{eq:adjoint_matrix_elements}.

An algebra is said to have \emph{adjoint-linear sparsity} when $s_{adj}=O(d)$; this scaling is characteristic of dense physical structures, including DLAs in the Pauli basis. A second sparsity parameter captures the structure of the Killing form factorization.

\begin{definition}[Interaction Sparsity ($s_{int}$)]\label{def:interaction_sparsity}
The padded interaction sparsity $s_{int}$ is a supplied bound on the number of basis elements $B_{i}$ that can combine with a fixed basis element $B_{k}$ (as the second argument) to produce a component along a fixed basis element $B_{m}$ (as the result),
\begin{equation}\label{eq:interaction_sparsity}
s_{int} \geq \max_{k,m} |\{i : f_{ikm} \neq 0\}|.
\end{equation}
\end{definition}

At $s_{int}=O(1)$, fixing the output $B_m$ and one input $B_k$ leaves only $O(1)$ possibilities for the other input $B_i$. This parameter determines the sparsity of the matrix $A$ in the Killing form factorization $K=A^T B$ (Section~\ref{sec:killing_form_factorization}), which is distinct from the sparsity of $B$ (governed by $s_{adj}$). The Pauli basis achieves $s_{int}=O(1)$ (as shown in Section~\ref{sec:performance_classes}), whereas generic bases typically have $s_{int}=O(d)$.

Finally, we introduce a definition to characterize the distribution of the magnitudes of the structure constants across the basis.

\begin{definition}[$\beta$-Balance]\label{def:beta_balance}
Let $N_{k} = \sum_{i,j} |f_{ijk}|^{2}$ be the squared $L_{2}$ norm of the structure constants associated with the output basis element $B_{k}$. The average squared norm is $\overline{N} = (1/d)\sum_{k} N_{k}$. For a nonzero structure tensor, define $\beta=(\max_{k} N_{k})/{\overline{N}}$ and call the representation $\beta$-balanced. For an abelian algebra all $N_k$ vanish, so this ratio is undefined; both projector targets are zero and are handled by the zero-target branch of the algorithm contract.
\end{definition}

The condition $\beta=O(1)$ is a one-sided upper-balance statement. It says that no output norm is much larger than the average, but it does not lower-bound the smallest nonzero $N_k$. The dependence of $\beta$ on subalgebra structure is treated in Section~\ref{sec:structural_dependencies_balance_conditioning}.

\subsection{The Quantum Linear Algebra Toolkit}
\paragraph{Typed sparse-access block encodings.}
A unitary $U$ is an $(\alpha,a,\eta)$ block encoding of a matrix $M$ when its
designated input-output block differs from $M/\alpha$ by at most $\eta/\alpha$ in
operator norm; in particular, $\alpha\geq\|M\|$. For a rectangular matrix, let
$J_R$ and $J_C$ embed its row and column spaces into power-of-two registers and
encode $\widehat M=J_RMJ_C^\dagger$ between the two valid-index signal projectors.
The right-kernel output is restricted back by $J_C^\dagger(\cdot)J_C$, so zero
padding does not add spurious kernel directions. For
$M\in\mathbb F^{R_M\times C_M}$, let $s_r(M)$ and $s_c(M)$ bound the number of
nonzeros in one row and one column, and choose
$\mu_M\geq\max_{r,c}|M_{rc}|$.

For the row interface, enlarge the column-index register by disjoint seed labels
$\lambda^r_1,\ldots,\lambda^r_{s_r(M)}$ and distinct sentinels
$\bot^r_1,\ldots,\bot^r_{s_r(M)}$. For each row $r$, let $\pi_r$ be a
permutation of this padded register satisfying
$\pi_r(\lambda^r_\ell)=c_M(r,\ell)$ on an active slot and
$\pi_r(\lambda^r_\ell)=\bot^r_\ell$ otherwise. The active outputs in a row are
distinct, so this partial injection extends to a permutation. Define
$\tau_c$ in the same way on a padded row-index register, using column seeds
$\lambda^c_\ell$ and distinct sentinels $\bot^c_\ell$. The reversible support
oracles act in place:
\begin{align}
O_r^M\ket{r,x}&=\ket{r,\pi_r(x)},&
O_c^M\ket{x,c}&=\ket{\tau_c(x),c},
\label{eq:typed_support_oracles}
\end{align}
and a value oracle
\begin{equation}
O_{\mathrm{val}}^M\ket{r,c,0^b}
=\ket{r,c,\operatorname{enc}_b(M_{rc})}.
\label{eq:typed_value_oracle}
\end{equation}
The value oracle returns zero whenever either coordinate is outside the valid row
or column set, including every seed and sentinel. Thus inactive slots do not add
matrix entries, and a zero row or zero column never causes division by a slice
norm.

Uniform superpositions over the seed labels, followed by the in-place support
permutations, give the typed preparation maps
\begin{align}
\mathsf{PREP}_r^M\ket{r,0}
&=\frac{1}{\sqrt{s_r(M)}}\sum_{\ell=1}^{s_r(M)}
 \ket{r,\pi_r(\lambda^r_\ell)},\nonumber\\
\mathsf{PREP}_c^M\ket{0,c}
&=\frac{1}{\sqrt{s_c(M)}}\sum_{\ell=1}^{s_c(M)}
 \ket{\tau_c(\lambda^c_\ell),c}.
\label{eq:typed_prepare_oracles}
\end{align}
No ordinal register remains in either prepared state: an actual nonzero at
$(r,c)$ therefore meets on the common coordinate pair $(r,c)$ even when its row
and column ordinals differ. The valid-index signal projectors exclude seeds and
sentinels. A SELECT step computes the $b$-bit magnitude and phase of
$M_{rc}/\mu_M$, applies the corresponding controlled rotation, and uncomputes the
value register. The standard sparse-access construction then gives an
\begin{equation}
\bigl(\alpha_M,a_M,\eta_M\bigr)\text{-block encoding},
\qquad
\alpha_M=\mu_M\sqrt{s_r(M)s_c(M)},
\label{eq:sparse_normalization_contract}
\end{equation}
with $O(1)$ calls to the support and value oracles and reversible-gate overhead polynomial in $\log(R_MC_M)$, $b$, and $\log(1/\eta_M)$~\cite[Lemma~48]{Gilyen2019}. The normalization is conservative because
$\|M\|_2\leq\sqrt{\|M\|_1\|M\|_\infty}\leq\mu_M\sqrt{s_r(M)s_c(M)}$.
Products multiply the normalizations and propagate block-encoding error in the usual weighted sum.

\paragraph{QSVT with a supplied gap.}
Let the chosen target matrix be $M$ and suppose that its singular values lie in
$\{0\}\cup[\Delta,\|M\|]$ for a supplied $\Delta>0$. If $U_M$ has normalization
$\alpha$, the normalized gap is $\delta=\Delta/\alpha$. An even polynomial that is
$\epsilon$-close to one at zero and $\epsilon$-close to zero on
$[\delta,1]$ has degree $O(\delta^{-1}\log(1/\epsilon))$. QSVT applies this
polynomial to the right singular vectors and returns an approximate block encoding
of the orthogonal projector onto $\ker M$ with the same asymptotic number of calls
to $U_M$ and $U_M^\dagger$~\cite{Gilyen2019}. If the target encoding has
unscaled error $\eta_{\mathrm{tar}}$, the stability bound for singular-value transformation adds
at most $4D\sqrt{\eta_{\mathrm{tar}}/\alpha}$ to the implemented polynomial at
degree $D$~\cite{Gilyen2019}. We therefore use
\[
\eta_{\mathrm{tar}}\leq
\frac{\alpha\epsilon^2}{256D^2},
\]
allocate at most $\epsilon/2$ to the ideal filter approximation and at most
$\epsilon/4$ to phase synthesis. The encoding perturbation is then at most
$\epsilon/4$. The theorem below assumes $\Delta$; it does not infer the gap from
singular-value estimation.

\begin{definition}[Spectral conditioning and block-encoding slack]\label{def:condition_number}
For a nonzero target $M$, write
\begin{equation}
\kappa_{\mathrm{spec}}(M)=\frac{\sigma_{\max}(M)}{\sigma_{\min}(M)},
\qquad
\rho_{\mathrm{BE}}(M)=\frac{\alpha}{\sigma_{\max}(M)}\geq 1,
\label{eq:kappa_rho_contract}
\end{equation}
where $\sigma_{\min}(M)$ is the least nonzero singular value and $\alpha$ is the
normalization of the implemented block encoding.
\end{definition}

\begin{unnumberedremark}[Runtime-factor contract]
With the tight promise
$\Delta=\sigma_{\min}(M)$, the runtime factor is exactly
\begin{equation}
\frac{\alpha}{\Delta}
=\rho_{\mathrm{BE}}(M)\,\kappa_{\mathrm{spec}}(M).
\label{eq:runtime_factor_contract}
\end{equation}
A smaller certified lower bound $\Delta$ gives the corresponding conservative
upper bound. If $M=0$, the target projector is the identity and no gap promise is
needed.
\end{unnumberedremark}

\section{Quantum Algorithm for Radical Identification}

The algorithm rests on the following standard structural fact, which converts the algebraic decomposition problem into a linear-algebraic kernel problem.

\begin{theorem}[Radical as Killing-orthogonal complement]\label{thm:radical_orthogonal_complement}
Let $\g$ be a finite dimensional Lie algebra over a field of characteristic zero. The radical $\rfrak = \Rad(\g)$ is the orthogonal complement of the derived algebra $\g^{(1)} = [\g, \g]$ with respect to the Killing form $\K$,
\begin{equation}\label{eq:radical_orthogonal_complement}
\rfrak = ([\g, \g])^{\perp_{\K}}.
\end{equation}
\end{theorem}
This is a standard consequence of Cartan's criteria~\cite{humphreys2012introduction}. We exploit it to reformulate radical identification as a kernel-finding task; the remainder of this section details the construction of the target operator, the bilinear factorization of the Killing form, and the QSVT-based kernel projection.

\subsection{The Radical as an Operator Kernel}\label{sec:radical_as_kernel}

The core requirement for an element $x\in\g$ to be in the radical $\rfrak$ is that it must be $\mathcal{K}$-orthogonal to every element $y$ in the derived algebra $[\g,\g]$. To translate this into linear algebra, we first characterize the derived algebra as the image of a specific linear operator.

We define the commutator matrix $C$ as a $d\times d^{2}$ matrix whose columns span the derived algebra. The rows are indexed by $k$, and the columns are indexed by the pair $(i, j)$. The entries of $C$ are the structure constants, indexed such that $C_{k,(i,j)}=f_{ijk}$. The image of $C$, $\ImC$, is precisely the vector subspace $[\g, \g]$.

The Killing form $\K$ is symmetric and bilinear, and the standard orthogonality for such a form, even over the complex field $\mathbb{C}$, is defined via the transpose, without complex conjugation:

\begin{equation}\label{eq:orthogonality_condition}
x^{T} K y = 0 \quad \forall y \in \ImC.
\end{equation}

Equation~\ref{eq:orthogonality_condition} makes the vector $x^{T}K$ orthogonal to all columns of $C$, so $x^{T}KC=0$. Transposing yields $C^{T}K^{T}x=0$, and the symmetry $K^{T}=K$ of the Killing form reduces this to $C^{T}Kx=0$. Thus the general target operator is $A_{\mathrm{gen}}=C^{T}K$; the task is to find its kernel.

\subsection{Block-Encoding Construction and Kernel Projection}\label{sec:workflow}

The target block encoding and its kernel projector are constructed in four steps.

\paragraph{Step 1: Input assumptions and pre-processing.}
The input fixes $\mathbb F\in\{\mathbb R,\mathbb C\}$, an ordered basis, an exact
zero pattern, and $b$-bit fixed-point values for every nonzero $f_{ijk}$, promised
to satisfy antisymmetry and Jacobi exactly in the represented field. These encoded
values define the algebra presented to the algorithm; no additional unencoded
algebra enters the statement. The input also selects one of two
targets:
\begin{equation}
A_{\mathrm{gen}}=C^{T}K=C^{T}A^{T}B
\quad\text{or}\quad
A_{\mathrm{cmp}}=K=A^{T}B,
\label{eq:two_route_targets}
\end{equation}
where the second target requires the promise that $\g$ is compact. A dense tensor can
be indexed for the typed oracles in $O(d^3)$ setup time and storage; an already sparse
list takes $O(\operatorname{nnz}(f))$ data insertion plus its row and column indices.

\paragraph{Step 2: Quantum unitary construction.}
The core of the quantum computation builds block encodings of the components of $A_{\mathrm{gen}}=C^{T}K$.

\paragraph{Killing Form Factorization.}\label{sec:killing_form_factorization}

A direct construction of the Killing form matrix $K$ is inefficient. To enable efficient quantum block encoding, we employ the bilinear factorization tailored to the expression $K_{ij}=\sum_{k,m}f_{ikm}f_{jmk}$ (Equation~\ref{eq:killing_form_matrix_elements}). This factorization is also central to the optimized classical algorithm.

We define two auxiliary matrices $A, B\in\mathbb{F}^{d^{2}\times d}$ (where $\mathbb{F}$ is the underlying field, $\mathbb{R}$ or $\mathbb{C}$). The rows of these matrices are indexed by the pair $(k, m)$, corresponding to the summation indices in Equation~\ref{eq:killing_form_matrix_elements}. The columns are indexed by a single index $i$ or $j$. We define the elements of $A$ and $B$ as follows:
\begin{equation}\label{eq:AB_definitions}
A_{(k,m),i}:=f_{ikm} \quad \text{and} \quad B_{(k,m),j}:=f_{jmk}.
\end{equation}
With these definitions, the Killing form matrix is precisely realized as the product $K=A^{T}B$. We verify this by examining the $(i, j)$-th element:
\[
(A^{T}B)_{ij}=\sum_{k,m}(A^{T})_{i,(k,m)}B_{(k,m),j}=\sum_{k,m}A_{(k,m),i}B_{(k,m),j}.
\]
Substituting the definitions from Equation~\ref{eq:AB_definitions}:
\[
(A^{T}B)_{ij}=\sum_{k,m}f_{ikm}f_{jmk}=K_{ij}.
\]
The identity above reduces a dense encoding of $K$ to composed encodings of the potentially sparser matrices $A^{T}$ and $B$. The relevant sparsities of $A$ and $B$ are captured by $s_{int}$ and $s_{adj}$, respectively. A standard sparse-access block encoding~\cite{Gilyen2019} already uses a constant number of oracle queries, while its subnormalization depends on the sparsity and entry magnitudes of the encoded matrix. Composing encodings of $A^{T}$ and $B$ therefore controls the implemented subnormalization that a direct dense encoding of $K$ would inflate.

\paragraph{Constituent Unitaries and Transpose Implementation}

The three matrices use the same typed interface but have different row and column
registers. Choose the certified padded bounds
\begin{equation}
s_{\mathrm{br}}\geq\max_{i,j}\bigl|\{k:f_{ijk}\neq0\}\bigr|,
\qquad
s_{\mathrm{out}}\geq\max_k\bigl|\{(i,j):f_{ijk}\neq0\}\bigr|.
\label{eq:commutator_sparsities}
\end{equation}
Their exact types and sparse-access bounds are
\begin{center}
\begin{tabular}{@{}lllll@{}}
\toprule
$M$ & dimensions & $M_{rc}$ & $s_r(M)$ & $s_c(M)$\\
\midrule
$B$ & $d^2\times d$ & $B_{(k,m),j}=f_{jmk}$ & $s_{int}$ & $s_{adj}$\\
$A^T$ & $d\times d^2$ & $(A^T)_{i,(k,m)}=f_{ikm}$ & $s_{adj}$ & $s_{int}$\\
$C^T$ & $d^2\times d$ & $(C^T)_{(i,j),k}=f_{ijk}$ & $s_{\mathrm{br}}$ & $s_{\mathrm{out}}$\\
\bottomrule
\end{tabular}
\end{center}
The row and column PREPARE maps of Equation~\eqref{eq:typed_prepare_oracles} enumerate
these supports, while SELECT obtains the corresponding $f_{ijk}$ value and phase.
Dummy slots encode zero. Setting $\mu_M=F$ in
Equation~\eqref{eq:sparse_normalization_contract} gives
\begin{align}
\alpha_B&=F\sqrt{s_{int}s_{adj}},&
\alpha_{A^T}&=F\sqrt{s_{adj}s_{int}},&
\alpha_{C^T}&=F\sqrt{s_{\mathrm{br}}s_{\mathrm{out}}}.
\label{eq:constituent_normalizations}
\end{align}
If $F=0$, every target is zero and the right-kernel projector is the identity;
the displayed sparse encodings are used only for $F>0$.
These are valid upper bounds for every encoded tensor; no slice-norm tightness is
assumed. The composite normalizations are therefore
\begin{align}
\alpha_{\mathrm{gen}}
&=\alpha_{C^T}\alpha_{A^T}\alpha_B
=F^3s_{adj}s_{int}\sqrt{s_{\mathrm{br}}s_{\mathrm{out}}},\nonumber\\
\alpha_{\mathrm{cmp}}
&=\alpha_{A^T}\alpha_B
=F^2s_{adj}s_{int}.
\label{eq:target_normalizations}
\end{align}
For either target, write
$D=O((\alpha/\Delta)\log(1/\epsilon))$ for the filter degree. The constituent
rotation errors are chosen, using the weighted product bounds of
Equation~\eqref{eq:composite_block_errors}, so that the composite target error
satisfies
$\eta_{\mathrm{tar}}\leq\alpha\epsilon^2/(256D^2)$. The filter and phase errors
use the remaining tolerance as specified above. The oracle for $A^T$ is
built by swapping the tensor indices and preserving the encoded complex phase; it
is the algebraic transpose and is not obtained by taking the adjoint of a block
encoding of $A$.

\paragraph{Step 3: Supplied spectral promise and phase synthesis.}\label{sec:qsve_step}
The input supplies $\Delta_{\mathrm{gen}}$ or $\Delta_{\mathrm{cmp}}$ as a certified
$b_{\Delta}$-bit fixed-point lower bound on the target's nonzero singular values. The QSVT phase angles are
computed for the normalized threshold $\Delta/\alpha$. No adaptive spectral search,
rank inference, or overlap assumption is included in this algorithm.
For a filter of degree $D$, the polynomial and its phase sequence are computed once
to per-angle precision $\xi=\Theta(\epsilon/D)$ in classical time
$\operatorname{poly}(D,b,b_{\Delta},\log(D/\epsilon))$. The resulting circuit
description is cached for fixed $\alpha$, $\Delta$, and $\epsilon$. This classical
setup is separate from the oracle-query count; every execution still pays the
queries and reversible gates stated below.

\paragraph{Step 4: QSVT kernel projection.}
QSVT applies an even kernel-filter polynomial of degree
$O((\alpha/\Delta)\log(1/\epsilon))$. The output unitary is a
$(1,a_P,\epsilon)$ block encoding of the Euclidean orthogonal projector onto
$\ker A_{\mathrm{gen}}=\rfrak$ in the general case, or onto
$\ker A_{\mathrm{cmp}}=\ker K=\zfrak=\rfrak$ in the compact case. Its signal block
acts linearly on coefficient-state superpositions. Applying the output unitary to a state and postselecting its
signal ancilla has probability $\|\widetilde P_{\rfrak}\ket{\psi}\|^2$ for the
normalized input $\ket{\psi}$. This probability differs from the ideal kernel
overlap $\|P_{\rfrak}\ket{\psi}\|^2$ by at most $2\epsilon$; the theorem makes no
uniform lower-bound claim for this separate use.

\subsection{Efficient Block Encoding via Arithmetic Oracles (Pauli Basis)}\label{sec:arithmetic_oracles}

The sparse-access construction combines entry or sign queries with the reversible
support enumeration in Equation~\eqref{eq:typed_prepare_oracles}. The full Pauli
algebra admits such an enumeration directly. A smaller Pauli-spanned family uses
an additional rank/select promise for its actual label set.

\paragraph{Algorithmic structure.}
Let $S=\mathbb F_2^{2n}\setminus\{0\}$ index the standard basis of
$\su(2^n)$ by $B_v=iP_v$. The commutator vanishes exactly when
$\omega(v,w)=0$ and otherwise has one output label $v\oplus w$, with coefficient
in $\{\pm2\}$~\cite{Gottesman1998}. For each nonzero $v$, choose a pivot coordinate
on which the nonzero linear functional $w\mapsto\omega(v,w)$ has coefficient one.
Given $u\in\mathbb F_2^{2n-1}$, fill the nonpivot coordinates of $w$ from $u$ and
solve the pivot bit so that $\omega(v,w)=1$. This reversible map is a bijection
\begin{equation}
E_v:\mathbb F_2^{2n-1}\longrightarrow
\{w\in\mathbb F_2^{2n}:\omega(v,w)=1\}.
\label{eq:pauli_affine_enumerator}
\end{equation}
A parity tree implements it with $O(n)$ gates and depth $O(\log n)$.

\paragraph{Oracle construction.}
The map $E_v$ enumerates a column support of $A$ or $B$. For fixed input and
output labels $(k,m)$, the only possible remaining label is $i=k\oplus m$, so the
corresponding row support has size at most one. For a fixed output $k$, the ordered
pairs contributing to $C$ are enumerated as
$(i,j)=(E_k(u),E_k(u)\oplus k)$. These maps give the row and column support oracles
for every matrix in the table above. Uniform superposition over $u$ supplies the
PREPARE states.

SELECT evaluates the Pauli phase law
$P_vP_w=i^{\phi(v,w)}P_{v\oplus w}$. The binary quadratic-form formula of
Dehaene and De Moor is evaluated reversibly to compute
$\phi(v,w)$ and write the coefficient sign~\cite{Dehaene2003}; the work registers
are then uncomputed. This implementation has size $O(n)$ and parallel depth
$O(\log n)$. Thus the full Pauli algebra uses no stored
$d^3$ tensor; one typed oracle call costs $O(n)$ reversible gates.

For an arbitrary explicit Pauli label set $S$, the support is
$\{w\in S:\omega(v,w)=1,\ v\oplus w\in S\}$. Membership tests do not provide a
reversible rank/select enumeration of this set. Such a family therefore falls under
the theorem only if this enumeration is supplied and costed, or if its supports are
loaded explicitly. Merely knowing the entry and sign rule does not justify an
implicit-oracle runtime claim.

\section{Complexity and Performance Analysis}\label{sec:complexity_analysis}

The quantum construction yields oracle-query bounds for each target. A separate
classical benchmark reports arithmetic-operation bounds for its explicit-basis
output, so the two cost units remain distinct.

\subsection{Quantum Complexity Analysis}

We first derive the quantum projector bound from the typed access contract.

\begin{theorem}[Promised-gap projector query complexity]\label{thm:general_complexity_main}
Let the $b$-bit structure tensor of a $d$-dimensional Lie algebra over
$\mathbb R$ or $\mathbb C$ be given through the typed support and value oracles of
Equations~\eqref{eq:typed_support_oracles}--\eqref{eq:typed_value_oracle}, together
with the certified magnitude and padded-support metadata above. Suppose
the nonzero singular values of $A_{\mathrm{gen}}=C^TA^TB$ are at least
$\Delta_{\mathrm{gen}}>0$, where $\Delta_{\mathrm{gen}}$ is supplied as a certified
$b_{\Delta}$-bit fixed-point number. For $\epsilon\in(0,1/2)$, there is a deterministic
quantum circuit that is a $(1,a_P,\epsilon)$ block encoding of
$P_{\ker A_{\mathrm{gen}}}=P_{\rfrak}$ and uses
\begin{equation}\label{eq:general_complexity}
T_Q^{\mathrm{gen}}
=O\!\left(
\frac{F^3s_{adj}s_{int}\sqrt{s_{\mathrm{br}}s_{\mathrm{out}}}}
{\Delta_{\mathrm{gen}}}
\,\log\frac1\epsilon
\right)
\end{equation}
calls to the constituent support and value oracles. With
$D_{\mathrm{gen}}=O((\alpha_{\mathrm{gen}}/\Delta_{\mathrm{gen}})
\log(1/\epsilon))$, the reversible-gate count is
\[
O\!\left(T_Q^{\mathrm{gen}}
\left[g_{\mathrm{acc}}+
\operatorname{poly}\!\left(\log d,b,b_{\Delta},
\log\frac{D_{\mathrm{gen}}}{\epsilon}\right)\right]\right).
\]
\end{theorem}

\begin{proof}
Equation~\eqref{eq:constituent_normalizations} follows from the sparse-access
block-encoding construction and
$\|M\|_2\leq\sqrt{\|M\|_1\|M\|_\infty}$. Multiplying the three encodings gives
the normalization $\alpha_{\mathrm{gen}}$ in
Equation~\eqref{eq:target_normalizations}. The promised normalized gap is
$\Delta_{\mathrm{gen}}/\alpha_{\mathrm{gen}}$, so the even QSVT kernel filter has
the degree in Equation~\eqref{eq:general_complexity}. The structural identity
$\ker(C^TK)=\rfrak$ identifies the resulting right-singular-space projector.
\end{proof}

For a nonzero target, the theorem also gives its tight-promise specialization.
The access contract alone gives $\rho_{\mathrm{BE}}\geq1$ and, from the gap
promise, the conservative upper bound
$\rho_{\mathrm{BE}}\leq\alpha_{\mathrm{gen}}/\Delta_{\mathrm{gen}}$.
No dimension-independent upper bound follows for an arbitrary basis.

The full Pauli algebra is a scoped exact case. Put
$q=2^{2n-1}=(d+1)/2$. Direct support counting gives
$F=2$, $s_{adj}=s_{\mathrm{out}}=q$, and
$s_{int}=s_{\mathrm{br}}=1$. Disjoint output supports give
$C C^T=4qI$, while the Killing form is $K=-4qI$. Hence
$\alpha_{\mathrm{cmp}}=4q=\sigma_{\max}(K)$ and
$\alpha_{\mathrm{gen}}=8q^{3/2}=\sigma_{\max}(C^TK)$.
Thus $\rho_{\mathrm{BE}}=1$ for both targets in this named family. This conclusion
is not extended to an arbitrary Pauli-spanned subalgebra.

\paragraph{Optimization for compact algebras.}

For a compact DLA ($\g = \s \oplus \zfrak$), identifying the radical is equivalent to identifying the kernel of the Killing form itself: $x \in \zfrak \iff x \in \ker(K)$.
    
\begin{corollary}[Optimized Complexity for Compact Algebras]\label{cor:compact_optimization}
Let $\g$ be compact and suppose the nonzero singular values of
$A_{\mathrm{cmp}}=K=A^TB$ are at least $\Delta_{\mathrm{cmp}}>0$, supplied as a
certified $b_{\Delta}$-bit fixed-point number. Under the same typed access and
precision contract, QSVT returns a
$(1,a_P,\epsilon)$ block encoding of $P_{\ker K}=P_{\zfrak}=P_{\rfrak}$ with
\begin{equation}
T_Q^{\mathrm{cmp}}
=O\!\left(
\frac{F^2s_{adj}s_{int}}{\Delta_{\mathrm{cmp}}}
\,\log\frac1\epsilon
\right).
\label{eq:compact_complexity_contract}
\end{equation}
For the tight promise, the runtime factor is
$\rho_{\mathrm{BE}}(K)\kappa_{\mathrm{spec}}(K)$.
The factor $\rho_{\mathrm{BE}}(K)$ remains explicit except in a family for which
tightness has been proved.
With $D_{\mathrm{cmp}}=O((\alpha_{\mathrm{cmp}}/\Delta_{\mathrm{cmp}})
\log(1/\epsilon))$, its reversible-gate count is
\[
O\!\left(T_Q^{\mathrm{cmp}}
\left[g_{\mathrm{acc}}+
\operatorname{poly}\!\left(\log d,b,b_{\Delta},
\log\frac{D_{\mathrm{cmp}}}{\epsilon}\right)\right]\right).
\]
\end{corollary}

\begin{remark}[Task boundary]\label{rem:state_prep_bottleneck}
The theorem constructs one projector block-encoding circuit. Preparing a coefficient
state, estimating the kernel dimension, returning an explicit basis, and simulating
a physical Hamiltonian are different tasks. Their input and output costs are not
included in $T_Q^{\mathrm{gen}}$ or $T_Q^{\mathrm{cmp}}$.
\end{remark}

\subsection{Zero-center decision from an all-direction probe}
\label{sec:zero_center_decision}

This subsection adds an output protocol for a compact DLA. The convention for the
supplied algebra must be fixed in advance as $\uu(N)$, $\su(N)$, or a projective
algebra in which the scalar $iI$ has been removed. The decision concerns the
center inside that supplied algebra.

Let $P=P_{\zfrak}=P_{\ker K}$ and let $R=\operatorname{rank}P=\dim\zfrak$.
Write $J:\mathbb C^d\rightarrow\mathcal H_{\mathrm{reg}}$ for the isometric
embedding of the valid coefficient labels into the qubit register. For the
projector circuit $U_P$, define its valid signal block by
\begin{equation}
B=(\bra{0}^{\otimes a_P}\otimes J^\dagger)U_P
  (\ket{0}^{\otimes a_P}\otimes J),
\qquad
\|B-P\|\leq\eta,\qquad \|B\|\leq1.
\label{eq:valid_center_block}
\end{equation}
The good event checks both the zero signal ancilla and a valid output coefficient
label. This second check excludes leakage into padded labels.

The all-direction probe is
\begin{equation}
\ket{\Phi_d}
=\frac1{\sqrt d}\sum_{j=1}^{d}\ket{Jj}_A\ket{Jj}_R.
\label{eq:maximally_entangled_probe}
\end{equation}

\begin{lemma}[Center-rank overlap]\label{lem:center_rank_overlap}
For any $d\times d$ matrix $M$,
\[
\|(M\otimes I)\ket{\Phi_d}\|^2
=\frac1d\Tr(M^\dagger M).
\]
Consequently, the good probability $q=\|(B\otimes I)\ket{\Phi_d}\|^2$ obeys,
for $0\leq\eta\leq1$,
\begin{equation}
\frac Rd(1-\eta)^2
\leq q\leq
\frac Rd+\left(1-\frac Rd\right)\eta^2.
\label{eq:center_block_contraction_bounds}
\end{equation}
For an exact projector, $q=R/d$.
\end{lemma}

\begin{proof}
The trace identity follows by expanding the reference register in its orthogonal
label basis. On $\operatorname{ran}P$, Equation~\eqref{eq:valid_center_block}
gives $1-\eta\leq\|Bu\|\leq1$ for a unit vector $u$. On $\ker P$, it gives
$\|Bv\|\leq\eta$. Summing these squared column norms proves
Equation~\eqref{eq:center_block_contraction_bounds}.
\end{proof}

The trace probe places the center problem within the earlier line of
all-direction singularity testing and normalized projector-rank estimation
~\cite{ambainis2012variable,mcardle2026streamlined}. The theorem below supplies
the compact-DLA specialization, including the valid coefficient domain and the
approximation-error budget for the stated promise.

\begin{theorem}[Bounded-error zero-center decision]
\label{thm:zero_center_decision}
Fix $1\leq R_{\min}\leq d$ and
$0<\delta_{\mathrm{fail}}<1/2$. Suppose the compact projector promises of
Corollary~\ref{cor:compact_optimization} hold. Assume that a unitary $V_\Phi$
prepares Equation~\eqref{eq:maximally_entangled_probe} entirely inside the valid
space and that $V_\Phi^\dagger$, $U_P^\dagger$, the good reflection, the
initial-state reflection, and controlled versions of the required circuits are
available. Controlled typed-oracle calls are charged at constant query overhead;
any larger implementation cost must be reported separately. If
\[
\eta\leq\frac14\sqrt{\frac{R_{\min}}d},
\]
then the promise
\begin{equation}
H_0:R=0,\qquad H_1:R\geq R_{\min}
\label{eq:zero_center_promise}
\end{equation}
can be decided with success probability at least
$1-\delta_{\mathrm{fail}}$ using
\begin{equation}
N_P=O\!\left(
\sqrt{\frac d{R_{\min}}}
\log\frac1{\delta_{\mathrm{fail}}}
\right)
\label{eq:zero_center_projector_queries}
\end{equation}
calls to the projector and state-preparation circuits and their inverses. The
corresponding typed structure-oracle count is
\begin{equation}
O\!\left[
\frac{F^2s_{adj}s_{int}}{\Delta_{\mathrm{cmp}}}
\sqrt{\frac d{R_{\min}}}
\log\!\left(4\sqrt{\frac d{R_{\min}}}\right)
\log\frac1{\delta_{\mathrm{fail}}}
\right].
\label{eq:zero_center_typed_queries}
\end{equation}
\end{theorem}

\begin{proof}
Put $a_0=R_{\min}/d$. Lemma~\ref{lem:center_rank_overlap} gives
$q\leq a_0/16$ under $H_0$ and $q\geq9a_0/16$ under $H_1$. The matching
amplitude angles are separated by at least $\sqrt{a_0}/2$. Standard amplitude
estimation resolves this gap with constant success using
$O(a_0^{-1/2})$ Grover iterations~\cite[Theorem~12]{brassard2000quantum}.
Independent repetition and a median reduce the failure probability, giving
Equation~\eqref{eq:zero_center_projector_queries}. Substituting
$\eta=\sqrt{a_0}/4$ into
Equation~\eqref{eq:compact_complexity_contract} gives
Equation~\eqref{eq:zero_center_typed_queries}. Appendix~\ref{app:zero_center_details}
records the angle and perturbation bounds.
\end{proof}

\begin{corollary}[Unitary state-preparation error]
\label{cor:zero_center_unitary_prep}
Let the actual pure preparation, including all work registers, be within Euclidean
distance $\xi\leq\sqrt{R_{\min}/d}/8$ of the ideal state. If its output is valid
and its exact inverse is available, it has the query scaling stated in
Theorem~\ref{thm:zero_center_decision}.
\end{corollary}

\begin{proposition}[Postselected center state]\label{prop:center_state_output}
If $R>0$ and $B=P$, normalizing the good branch gives
\[
\ket{\Psi_P}
=\sqrt{\frac dR}(P\otimes I)\ket{\Phi_d}
=\frac{\operatorname{vec}(P)}{\sqrt R}.
\]
Its two reduced states are $P/R$ and $P^T/R$. The coefficient register is
maximally mixed on the center subspace; this output is not a classical basis or
a preferred pure center vector.
\end{proposition}

\begin{remark}[Rank and output boundary]\label{rem:zero_center_rank_boundary}
The same signal permits an additive $\varepsilon_p$ estimate of the normalized
rank $R/d$ with
$O(\varepsilon_p^{-1}\log(1/\delta_{\mathrm{fail}}))$ outer calls. Exact rank
has a different worst-case cost. The input-dependent expected count of
Brassard et al.\ is
$\Theta(\sqrt{(R+1)(d-R+1)})$~\cite[Theorem~17]{brassard2000quantum}, while
adjacent middle ranks require $\Omega(d)$ fixed-error queries in the Boolean
oracle model~\cite[Corollary~1.2]{NayakWu1999}. Thus the worst-case exact-rank
outer complexity is $\Theta(d)$. An approximate projector adds normalized-rank
bias at most $2\eta$, so worst-case exact recovery also needs
$\eta=O(1/d)$. These statements do not produce a center basis.
\end{remark}

\begin{definition}[Opaque typed-value access]\label{def:opaque_typed_value_access}
In the opaque typed-value model used below, the ordered basis, support, and
listed metadata are public, while numerical values are accessible only through
opaque typed-value unitaries and their inverses. No readable implementation or
value table is supplied. Setup data, global norms, and equivalent
answer-bearing statistics are also excluded.
\end{definition}

\begin{theorem}[Opaque-value oracle-model comparison]\label{thm:opaque_value_separation}
There is an explicit family of compact real Lie algebras
$\g_x=\bigoplus_{j=1}^{N}\mathfrak h_{x_j}$, where
\[
\mathfrak h_0=\su(2)\oplus\su(2),
\qquad
\mathfrak h_1=\su(2)\oplus\uu(1)^{\oplus3},
\]
with dimension $d=6N$ and center rank $R=3|x|$. The two six-dimensional
blocks have the same exact zero pattern in one common basis, and they share the
public bounds
\[
F=44,\quad
(s_{adj},s_{int},s_{\mathrm{br}},s_{\mathrm{out}})
=(30,5,6,30),\quad
\Delta_{\mathrm{cmp}}=\frac1{32}.
\]
In the model of Definition~\ref{def:opaque_typed_value_access}, under the promise
$|x|=0$ or $|x|=t$, where $1\leq t\leq N/2$, the fixed-error query
complexities satisfy
\begin{equation}
Q_{\mathrm{typed}}
=\Theta\!\left(\sqrt{\frac Nt}\right)
=\Theta\!\left(\sqrt{\frac d{R_{\min}}}\right),
\qquad
C_{\mathrm{typed}}
=\Theta\!\left(\frac Nt\right)
=\Theta\!\left(\frac d{R_{\min}}\right),
\label{eq:opaque_value_zero_center_separation}
\end{equation}
with $R_{\min}=3t$.
\end{theorem}

The compact-Lie certificates and complete query reduction are given in Appendix~\ref{app:zero_center_details}.

\paragraph*{Query-complexity provenance.}
Theorem~\ref{thm:opaque_value_separation} transfers the established
promised-search exponents~\cite{brassard2000quantum,NayakWu1999} through a
promise-preserving Lie-algebra construction. Both branches are compact Lie
algebras satisfying the Jacobi identity, and their structure tensors share the
public zero pattern and common sparsity, scale, and gap certificates while their
center ranks differ. The resulting witness isolates the typed-value oracle model.

\begin{remark}[Scope of the query separation]\label{rem:opaque_value_scope}
Theorem~\ref{thm:opaque_value_separation} does not apply to the explicit-input
classical benchmark of Section~\ref{sec:optimized_classical_benchmark}. If the
values or their setup are freely readable, a scan can expose or cache the
answer. In this family, a direct Grover search already reaches the quantum
upper bound, and the center projector itself can be selected with a constant
number of hidden-bit queries. The theorem bounds global decision extraction.
It gives no lower bound for QSVT, projector synthesis, data loading, gate
complexity, DLA generation, or a physical implementation.
\end{remark}

\subsection{Structural Dependencies: Balance and Conditioning in the Pauli Basis}\label{sec:structural_dependencies_balance_conditioning}

The conservative sparse-access normalization depends on
$F,s_{adj},s_{int},s_{\mathrm{br}},s_{\mathrm{out}}$. The balance parameter
$\beta$ remains relevant to the separate conditioning analysis but is not used to
claim a sharper block-encoding normalization. A Pauli basis gives $F=O(1)$ and
$s_{int}=O(1)$, yet it does not by itself supply support enumeration, balance, or
conditioning for an arbitrary Pauli-spanned DLA.

This analysis focuses on real compact DLAs, where $\g = \s \oplus \zfrak$ and $r := \dim(\s)$. The input basis for $\g$ is assumed to be a subset of the standard Skew-Hermitian Pauli operators, i.e.\ $\g$ is Pauli-spanned. The restriction buys two facts the analysis below leans on: the structure constants $f_{ijk}$ are real, and the deterministic Pauli product law of Section~\ref{sec:condition_number_barrier} applies.

\begin{table}[ht]
\centering
\caption{Structural parameters and conditioning quantities. The sparse-access
normalization uses explicit row and column sparsities; $\beta$ enters only the
separate conditioning discussion.}
\label{tab:complexity_parameters}
\footnotesize
\begin{tabular}{@{}l p{0.42\textwidth}@{}}
\toprule
Symbol & Role in the algorithmic contract \\
\midrule
$F$ & entry-magnitude bound in every constituent normalization \\
$s_{adj},s_{int}$ & row/column bounds for $A^T$ and $B$ \\
$s_{\mathrm{br}},s_{\mathrm{out}}$ & row/column bounds for $C^T$ in the general construction \\
$\beta$ & records $N_{\max}/\overline N$; does not determine $\kappa(D|_{\s})$ without lower-tail control and is not used to sharpen $\alpha$ \\
$\kappa(K|_{\s})$ & block-scalar ratio $\max_l N_l/\min_l N_l$ \\
$\kappa_{\mathrm{spec}}(A_{\mathrm{gen}})$ & $\kappa_{\mathrm{spec}}(K|_{\s})^{3/2}$ for $\s\ne0$ in the compact real invariant-orthonormal frame \\
$\kappa_{\mathrm{spec}}(A_{\mathrm{cmp}})$ & $\kappa_{\mathrm{spec}}(K)$ for $K\ne0$; equals $\kappa_{\mathrm{spec}}(K|_{\s})$ in an adapted invariant-orthonormal frame \\
\bottomrule
\end{tabular}
\end{table}

\subsubsection{Analysis of Balance $\beta$}\label{sec:analysis_balance_beta}

The balance parameter $\beta$ measures upper concentration of the structure-constant norms across the basis. We record its exact adapted-frame decomposition and its lower-tail limitation.

\begin{theorem}[Structural Dependence of $\beta$]\label{thm:beta_structure}
Let $\g=\zfrak\oplus\bigoplus_l\s_l$ be a compact real Lie algebra with nonzero semisimple part, with $d_l=\dim\s_l$, $r=\sum_l d_l>0$, and $d=\dim\g$. Assume the coefficient basis is invariant-orthonormal and adapted to this splitting, and write $K|_{\s_l}=-N_lI$. Define $N_{\max}=\max_lN_l$ and $N_{\min}=\min_lN_l$. Then
\begin{equation}
\beta=\frac{dN_{\max}}{\sum_l d_lN_l}
=\beta_{\mathfrak{s}}\frac{d}{r},
\qquad
\beta_{\mathfrak{s}}=\frac{rN_{\max}}{\sum_l d_lN_l}.
\end{equation}
In particular, a simple algebra has $\beta=1$.
\end{theorem}

\begin{proof}
The adapted basis has $N_k=0$ on the center and $N_k=N_l$ for every $B_k\in\s_l$, since $CC^T=-K$. Hence $\sum_kN_k=\sum_ld_lN_l$ and $\max_kN_k=N_{\max}$. Substitution in Definition~\ref{def:beta_balance} gives both displayed formulas. For one simple ideal the numerator and denominator coincide.
\end{proof}

\paragraph{Implications for $\beta$ Scaling.}
Imbalance can enter through the distribution of the block scalars $N_l$ or through center domination, where $r\ll d$. The direct-sum identity alone supplies no typicality statement.

\begin{example}[Structural Imbalance $O(d)$ in the Pauli Basis]\label{ex:structural_imbalance_pauli}
Consider the compact DLA $\g = \su(2) \oplus \uu(1)^{\oplus m}$. This algebra can be represented entirely within the Pauli basis. Here, $\s=\su(2)$ and $\zfrak=\uu(1)^{\oplus m}$. The dimensions are $r=3$ and $d=3+m$. The semisimple part $\su(2)$ is intrinsically balanced in the Pauli basis ($\beta_{\mathfrak{s}}=1$).

Applying Theorem~\ref{thm:beta_structure}:
\[
\beta = 1 \cdot \frac{3+m}{3}.
\]
If $m=\Theta(d)$, then $\beta=\Theta(d)$. This is a center-dominated algebra with a fixed interacting core.
\end{example}

\subsubsection{Analysis of Conditioning $\kappa$}\label{sec:condition_number_barrier}

We turn to the condition number of the general composite target $A_{\mathrm{gen}}=C^T K$ for a compact DLA treated as a real Lie algebra. Whenever its real coefficient basis is orthonormal for a positive invariant inner product and adapted to the splitting $\g=\s\oplus\zfrak$, total antisymmetry of the real structure constants forces $CC^{T}=-K$. If $\s\ne0$, then
\[
\kappa(A_{\mathrm{gen}}) = \kappa(K|_{\s})^{3/2}.
\]
For $\s=0$, both matrices vanish and the nonzero-spectrum condition numbers are undefined.
Lemma~\ref{lem:invariant_reduction} of Appendix~\ref{app:invariant_reduction} proves this with no appeal to the Pauli product law or to Schur--Weyl, and Theorem~\ref{thm:obstruction_general} below obtains the right-hand side from the simple-ideal blocks. Complex skew-Hermitian matrix entries cause no problem because their real span is the coefficient algebra. A genuinely complex coefficient basis falls outside this reduction. Transpose is then bilinear, whereas singular values are governed by the Hermitian adjoint. The Pauli basis is a real invariant-orthonormal frame, and there $CC^{T}=-K$ is diagonal and directly computable.

This $3/2$ law does not set the runtime factor of the compact projector algorithm.
The compact algorithm uses $A_{\mathrm{cmp}}=K$. When $\s\ne0$, its target-level
spectral condition number is $\kappa_{\mathrm{spec}}(K)$, as stated in
Corollary~\ref{cor:compact_optimization}. In the Pauli basis assumed in this
section, the basis is adapted to $\s\oplus\zfrak$ and $K$ is diagonal, so
$\kappa_{\mathrm{spec}}(K)=\kappa_{\mathrm{spec}}(K|_{\s})$. The zero-target
branch applies when $\s=0$.

That Pauli specialization is made explicit in Appendix~\ref{app:pauli_proofs}, where diagonality of $CC^{T}$ (Lemma~\ref{lemma:CCT_diagonal}) puts the condition number in the closed form of Theorem~\ref{thm:kappa_simplification}. The reduction's consequence for the dimensions of the simple ideals is the following.

\begin{theorem}[General dimensional obstruction]\label{thm:obstruction_general}
Let $\g$ be a compact real Lie algebra with decomposition $\g=\s\oplus\zfrak$ and nonzero semisimple part $\s=\bigoplus_{l=1}^{m}\s_l$, where $m\ge1$ and the simple ideal $\s_l$ has dimension $d_l$. Let $\{B_k\}$ be a real coefficient basis orthonormal for a positive invariant inner product and adapted to $\g=\bigoplus_l\s_l\oplus\zfrak$. Here $C$, $K$, and $A_{\mathrm{gen}}=C^T K$ are real matrices, and $T$ is Euclidean transpose. If $h^{\vee}_l$ is the dual Coxeter number of $\s_l$ and $\nu_l$ is the squared basic length shared by its generators, then
\[
K|_{\s_l}=-N_l\,I_{d_l},\qquad N_l=2h^{\vee}_l\,\nu_l,
\]
and the nonzero-spectrum conditioning is
\[
\kappa(A_{\mathrm{gen}})=\kappa(K|_{\s})^{3/2}=\Bigl(\frac{\max_l N_l}{\min_l N_l}\Bigr)^{3/2}.
\]
The exponent $3/2$ is fixed for every such real invariant-orthonormal frame. The type and block normalization determine the relation between $N_l$ and $d_l$.
\end{theorem}

\begin{proof}
Appendix~\ref{app:invariant_reduction} supplies the three statements for the proof. Lemma~\ref{lem:invariant_reduction} turns invariance of the real inner product into total antisymmetry of the structure constants, hence $CC^{T}=-K$ and $\kappa(A_{\mathrm{gen}})=\kappa(K|_{\s})^{3/2}$. Corollary~\ref{cor:invariant_block_ratio} uses simplicity of each $\s_l$ to collapse $K|_{\s_l}$ to a single scalar $-N_l$, so $\kappa(K|_{\s})=\max_l N_l/\min_l N_l$. Proposition~\ref{prop:casimir_dictionary} supplies $N_l=2h^{\vee}_l\nu_l$ in the chosen normalization.
\end{proof}

The normalizations that recur through the paper now fall out as corollaries, set apart only by how each frame presents its generators; nothing internal to $\g$ distinguishes them.

\begin{corollary}[Pauli frame, type-$A$ family]\label{cor:simple_summand_obstruction}
Suppose each $\s_l$ is Pauli-spanned and isomorphic to $\su(2^{n_l})$, and assume that its Pauli labels are exactly the nonzero vectors of a $2n_l$-dimensional non-degenerate symplectic subspace of $\mathbb{F}_2^{2n}$. Then $d_l=4^{n_l}-1$, and in the Pauli basis
\[
N_k=2(d_l+1)\ \text{ for } k\in\s_l,\qquad K|_{\s_l}=-2(d_l+1)\,I_{d_l},
\]
and the two conditioning factors are
\[
\kappa(D|_{\s})=\sqrt{\tfrac{d_{\max}+1}{d_{\min}+1}}=\Theta\!\bigl(\sqrt{d_{\max}/d_{\min}}\bigr),
\qquad
\kappa(A_{\mathrm{gen}})=\Bigl(\tfrac{d_{\max}+1}{d_{\min}+1}\Bigr)^{3/2}=\Theta\!\bigl((d_{\max}/d_{\min})^{3/2}\bigr).
\]
So $\kappa(A_{\mathrm{gen}})=O(1)$ forces every simple summand to share one dimension up to a constant, and the saturating case $d_{\min}=O(1)$, $d_{\max}=\Theta(d)$, realized by $\s=\su(2)\oplus\su(2^{k})$, attains the Pauli ceiling $\kappa(A_{\mathrm{gen}})=\Theta(d^{3/2})$.
\end{corollary}

\begin{proof}
The type is $A_{2^{n_l}-1}$, whose dual Coxeter number is $h^{\vee}_l=2^{n_l}$ (Table~\ref{tab:dual_coxeter}). The skew-Hermitian Pauli generators are the full-rank, unit-operator-norm frame of Remark~\ref{rmk:type_dependent_scaling}, with squared basic-length $\nu_l=2^{n_l}=\sqrt{d_l+1}$, so Proposition~\ref{prop:casimir_dictionary} gives $N_l=2h^{\vee}_l\nu_l=2\cdot 4^{n_l}=2(d_l+1)$; the direct count of anti-commuting Pauli pairs over $\mathbb{F}_2^{2n_l}$ in Appendix~\ref{app:pauli_proofs} reaches the same constant. The conditioning then follows from Theorem~\ref{thm:obstruction_general}. In the Pauli basis $CC^{T}=D^{2}=-K$, so $\kappa(D|_{\s})=\kappa(K|_{\s})^{1/2}$ is the square root of the same ratio.
\end{proof}

\begin{remark}[Isotypic decompositions and scope]\label{rmk:simple_summand_obstruction}
The full centralizer of a collective $SU(2)$ action has simple ideals indexed by spin sectors, but an arbitrary symmetric control algebra can be a proper subalgebra of that centralizer. For $n$ qubits the largest multiplicity is $\Theta(2^n/n)$ at $J=\sqrt n/2+O(1)$. The even-$n$ singlet multiplicity is smaller,
\[
m_0=\frac{1}{n/2+1}\binom{n}{n/2}=\Theta\!\left(\frac{2^n}{n^{3/2}}\right).
\]
Theorem~\ref{thm:su2_symmetric_obstruction} therefore concerns the full centralizer. A DLA specified by a smaller generating set needs its own Levi decomposition before Theorem~\ref{thm:obstruction_general} can be applied.
\end{remark}

The Schur--Weyl multiplicities give a worked structural instance once both the full-centralizer hypothesis and the metric are fixed.

\begin{lemma}[Centralizer structure under collective $SU(2)$]\label{lem:collective_su2_centralizer}
Let $\mathcal H_n=(\C^{2})^{\otimes n}$ carry the diagonal $SU(2)$ action, and write
\[
\mathcal H_n=\bigoplus_{J\in\mathcal J_n}V_J\otimes M_J,
\qquad r_J=\dim V_J=2J+1,
\]
with multiplicity
\[
m_J=\dim M_J=\binom{n}{n/2-J}-\binom{n}{n/2-J-1}.
\]
The skew-Hermitian centralizer in $\uu(2^n)$ is
\[
\mathfrak c_{\uu}=\bigoplus_{J\in\mathcal J_n}I_{V_J}\otimes\uu(m_J),
\]
while its traceless part is
\[
\mathfrak c_{\su}=\left\{\bigoplus_J I_{V_J}\otimes X_J:
X_J\in\uu(m_J),\ \sum_J r_J\Tr(X_J)=0\right\}.
\]
Let $\mathcal I_n=\{J:m_J\ge2\}$ and $q=|\mathcal J_n|=\lfloor n/2\rfloor+1$. Then
\[
[\mathfrak c_{\su},\mathfrak c_{\su}]
=\bigoplus_{J\in\mathcal I_n}I_{V_J}\otimes\su(m_J),
\qquad \dim\mathfrak z(\mathfrak c_{\su})=q-1.
\]
\end{lemma}

\begin{theorem}[Full-centralizer conditioning]\label{thm:su2_symmetric_obstruction}
For $n\ge3$, choose a real invariant-orthonormal basis adapted to the
centralizer splitting of Lemma~\ref{lem:collective_su2_centralizer}, and let $\nu_J$ be its squared basic length on
the $J$-th simple ideal. The block scalar is $N_J=2m_J\nu_J$, and
\begin{equation}\label{eq:kappa_su2}
\kappa(A_{\mathrm{gen}})=
\left(\frac{\max_{J\in\mathcal I_n}m_J\nu_J}
{\min_{J\in\mathcal I_n}m_J\nu_J}\right)^{3/2}.
\end{equation}
For $n\le2$ the centralizer has no semisimple ideal, so this nonzero-spectrum condition number is undefined.
\end{theorem}

Appendix~\ref{app:invariant_reduction} carries the proof.

The metric changes the asymptotic law. Under the blockwise multiplicity-space metric $\langle I\otimes X,I\otimes Y\rangle_{\mathrm{blk}}=-\Tr_{M_J}(XY)/2$, a Gell--Mann frame has $\nu_J=2$ and $N_J=4m_J$. If $m_+=\max_Jm_J$ and $m_-^+=\min_{m_J\ge2}m_J$, then
\[
\kappa_{\mathrm{blk}}(A_{\mathrm{gen}})=\left(\frac{m_+}{m_-^+}\right)^{3/2}
\sim \left(\frac{8}{\pi e}\right)^{3/4}\frac{2^{3n/2}}{n^3}.
\]
Here $m_-^+=n-1$ for $n\ge5$, with $m_-^+=2$ at $n=3,4$, and $m_+\sim\sqrt{8/(\pi e)}\,2^n/n$. Under the ambient metric $\langle X,Y\rangle_{\mathrm{amb}}=-\Tr_{\mathcal H_n}(XY)/2$, the normalized generators have $\nu_J=2/r_J$ and $N_J=4m_J/r_J$. This gives
\[
\kappa_{\mathrm{amb}}(A_{\mathrm{gen}})=
\left(\frac{\max_{J\in\mathcal I_n}m_J/r_J}
{\min_{J\in\mathcal I_n}m_J/r_J}\right)^{3/2}
\sim \left(\frac{8}{\pi}\right)^{3/4}\frac{2^{3n/2}}{n^{9/4}}.
\]
These are algebraic conditioning statements. Neither metric supplies the typed support enumerators required by the algorithm contract.

\subsubsection{Bounds on $\kappa(D|_{\mathfrak{s}})$ and $\kappa(K|_{\mathfrak{s}})$}

Assume $\s\ne0$ throughout this subsection.

In the Pauli frame the reduction of Appendix~\ref{app:pauli_proofs} makes the conditioning concrete. It is fixed by the Killing form $K$ together with the norm distribution matrix $D$ on the semisimple part $\s$, and bounded by the product of the factor condition numbers:
\begin{equation}
\kappa(A_{\mathrm{gen}}) \le \kappa(K|_{\mathfrak{s}}) \kappa(D|_{\mathfrak{s}}).
\end{equation}

\paragraph{Bounds on $\kappa(D|_{\mathfrak{s}})$.}
The condition number of $D$ restricted to $\s$ depends on the full extremal ratio of the nonzero block norms.
\[
\kappa(D|_{\mathfrak{s}}) = \sqrt{\frac{\max_{k \in \mathfrak{s}} N_k}{\min_{k \in \mathfrak{s}} N_k}}.
\]
Thus $N_{\max}/N_{\min}=O(1)$ implies $\kappa(D|_{\s})=O(1)$, but $\beta_{\s}=O(1)$ alone does not. For example, two simple blocks with scalars $1$ and $L$ have $\beta_{\s}<2$ when their dimensions agree, while $\kappa(D|_{\s})=\sqrt L$. A naive worst case admits $N_{\max}=O(d^2)$ and $N_{\min}=O(1)$, which bounds $\kappa(D|_{\mathfrak{s}})$ only by $O(d)$. The Pauli regime is tighter. At fixed $i$ and $k$ the Pauli product law admits at most one $j$ with $B_iB_j\propto B_k$, so $N_k=\sum_{i,j}f_{ijk}^2\le d\cdot 4=4d$. Each $B_k\in\s$ receives a non-zero contribution from at least one pair, so $N_k\ge 4$ on $\s$. Combining the two bounds yields
\[
\kappa(D|_{\mathfrak{s}}) \;\le\; \sqrt{\tfrac{4d}{4}} \;=\; O(\sqrt d) \qquad \text{for any Pauli-spanned $\g$.}
\]
The generic $O(d)$ envelope is therefore not saturable inside the Pauli regime. This $O(\sqrt d)$ tightening is itself sharp: Corollary~\ref{cor:simple_summand_obstruction} attains it on the family $\s=\su(2)\oplus\su(2^{k})$.

\paragraph{Bounds on $\kappa(K|_{\mathfrak{s}})$.}
For simple $\s$, Schur's Lemma forces $K\propto$ Trace form; in the Pauli basis (Trace-form orthonormal) this gives $K|_{\mathfrak{s}}\propto I$ and $\kappa(K|_{\mathfrak{s}})=O(1)$. The general Pauli-spanned case is tighter than continuous bases.

\begin{theorem}[Polynomial Conditioning of Integer Lie Algebras]\label{thm:kappa_poly_bound}
Let $\g$ be a Lie algebra spanned by a subset of the Skew-Hermitian Pauli basis. The Killing form $K$ is an integer matrix with a strictly diagonal structure, and its condition number, taken over the nonzero spectrum, satisfies $\kappa(K|_{\s}) \le d$; exponential ill-conditioning is therefore mathematically excluded.
\end{theorem}
\begin{proof}
The structure constants in the Pauli basis are integers, $f_{ijk} \in \{0, \pm 2\}$. The matrix entries of the Killing form are $K_{ij} = \sum_{k,m} f_{ikm} f_{jmk}$.
Due to the property that the product of two Pauli operators is proportional to a unique third Pauli, $P_i P_k \propto P_m$, it follows that $K_{ij} = 0$ for $i \neq j$. The matrix $K$ is strictly diagonal.
The diagonal entries are given by $K_{ii} = -\sum_{k,m} (f_{ikm})^2$. Since $f_{ikm}$ are integers, $K_{ii}$ must be an integer multiple of 4. Specifically, for any index $i$ in the semisimple component (where $K_{ii} \neq 0$), we have $|K_{ii}| \ge 4$.
This integer spectral gap is critical. Unlike continuous bases where eigenvalues can approach zero arbitrarily closely, the discreteness of the Pauli group enforces a lower bound on the non-zero singular values.
For each fixed $i$ at most $d$ structure constants $f_{ikm}$ are nonzero (one $m$ per anticommuting $k$), each of magnitude $2$, so $\max_i |K_{ii}| \le 4d$.
The condition number is thus strictly bounded:
\[
\kappa(K|_{\s}) = \frac{\sigma_{\max}}{\sigma_{\min}} \le \frac{4d}{4} = d.
\]
Combined with $\kappa(D|_{\mathfrak{s}})\le O(\sqrt d)$ and Theorem~\ref{thm:obstruction_general}, this gives $\kappa(A_{\mathrm{gen}})\le O(d^{3/2})$ in the Pauli regime.
\end{proof}

\subsubsection{The Optimal Case: Well-Conditioning of $\su(2^{n})$}\label{sec:well_conditioning_su2n_section}

We now prove that $\su(2^n)$ is well-conditioned in the Pauli basis.

\begin{theorem}[Well-Conditioning of $\su(2^n)$ in the Pauli Basis]\label{thm:well_conditioning_su2n}
Let $\g = \su(2^n)$ be represented in the standard Skew-Hermitian Pauli basis. The operator $A_{\mathrm{gen}}=C^T K$ is well-conditioned, i.e., its condition number $\kappa=O(1)$.
\end{theorem}

A proof is given in Appendix~\ref{app:pauli_proofs}.

\subsection{Structure-Aware Classical Benchmark}\label{sec:optimized_classical_benchmark}

We count a direct classical calculation under the same ordered basis and encoded structure
tensor. The de~Graaf literature supplies the algebraic method, while the exponents are
proved below~\cite{deGraaf2000,deGraaf1997Levi}. The primary classical output is an
explicit radical basis, so this theorem is not a like-for-like runtime comparison
with the quantum projector block encoding.

\begin{theorem}[Structure-Aware Classical Benchmark]\label{thm:classical_benchmark_main}
Given the explicit $b$-bit structure data in indexed sparse form, a direct
unit-cost field-operation algorithm returns an explicit basis of the radical within
\begin{align}
T_{C,\mathrm{gen}}
&=O(d\,s_{adj}s_{int}+d^4)
&&\text{for a general Lie algebra},\nonumber\\
T_{C,\mathrm{cmp}}
&=O(d\,s_{adj}s_{int}+d^3)
&&\text{for a compact Lie algebra}.
\label{eq:classical_two_route_bounds}
\end{align}
The field-operation statement treats the supplied dyadic real or complex numbers as
exact elements. If all entries share a $b$-bit denominator and are cleared to
integers, fraction-free elimination gives the conservative bit bounds
\begin{align}
\widetilde O\!\bigl(d^5(b+\log d)\bigr)
&\quad\text{for the general case},\nonumber\\
\widetilde O\!\bigl(d^4(b+\log d)\bigr)
&\quad\text{for the compact case},
\label{eq:classical_bit_bounds}
\end{align}
in addition to the sparse construction work. These are exact-arithmetic bounds;
no floating-point rank-stability theorem is asserted.
\end{theorem}

The complete operation and bit counts are given in Appendix~\ref{app:classical_benchmark}.

\paragraph{Classical center of an actual Pauli label set.}
Let $S\subseteq\mathbb F_2^{2n}$ be the distinct labels whose Pauli operators form
the given real basis, and let $V=\operatorname{span}_{\mathbb F_2}S$. Closure under
the bracket is assumed. Then
\begin{equation}
\zfrak
=\operatorname{span}_{\mathbb R}
\{B_v:v\in S\cap V^{\perp_\omega}\}.
\label{eq:pauli_center_actual_labels}
\end{equation}
Indeed, for fixed $w\in S$, the output labels $v\oplus w$ are distinct as $v$
varies, so noncommuting terms cannot cancel in
$[\sum_v a_vB_v,B_w]$. Row reduction of the $d\times2n$ label matrix produces a
basis of $V$; testing every actual $v\in S$ against that basis costs
$O(dn^2)$ bit operations, or $O(dn)$ word operations with $\Theta(n)$-bit packed
words, plus the output cost. It is $\widetilde O(d)$ only under an additional
$n=O(\operatorname{polylog}d)$ regime.

The intersection with the actual label set is essential. On two qubits,
$S=\{x_1,x_2,x_1+x_2\}$ labels three commuting Pauli operators, so the real center
has dimension three. The binary space $V$ has dimension two. Thus replacing
Equation~\eqref{eq:pauli_center_actual_labels} by the binary radical
$V\cap V^{\perp_\omega}$ loses the real-basis multiplicity even in this smallest
check.

\subsection{Performance Summary}\label{sec:performance_classes}

Table~\ref{tab:performance} separates the input, output, and cost unit for each
case. This organization keeps the projector-block-encoding, explicit-basis, and
decision tasks distinct.

\begin{table}[H]
\centering
\caption{Inputs, outputs, and distinct counting units for the principal cases.}
\label{tab:performance}
\normalsize
\setlength{\tabcolsep}{3pt}
\renewcommand{\arraystretch}{1.05}
\begin{tabularx}{\textwidth}{@{}
  >{\raggedright\arraybackslash}p{0.12\textwidth}
  >{\raggedright\arraybackslash}p{0.075\textwidth}
  >{\raggedright\arraybackslash}X
  >{\raggedright\arraybackslash}X
  >{\raggedright\arraybackslash}p{0.17\textwidth}@{}}
\toprule
Case & Target & Quantum: output; cost & Classical: output; cost & Status\\
\midrule
General (explicit)
& $C^TK$
& projector block encoding; $O((\alpha_{\mathrm{gen}}/\Delta_{\mathrm{gen}})\log(1/\epsilon))$ oracle calls
& explicit radical basis; $O(d\,s_{adj}s_{int}+d^4)$ field operations
& no ratio claimed\\
Compact (explicit)
& $K$
& projector block encoding; $O((\alpha_{\mathrm{cmp}}/\Delta_{\mathrm{cmp}})\log(1/\epsilon))$ oracle calls
& explicit radical basis; $O(d\,s_{adj}s_{int}+d^3)$ field operations
& no ratio claimed\\
Compact zero-center wrapper
& $P_{\ker K}$ signal
& bounded-error bit; Eq.~\eqref{eq:zero_center_typed_queries}
& decision from the explicit compact calculation above
& extra control contract; no ratio claimed\\
Opaque same-support family
& zero-center decision
& $\Theta(\sqrt{d/R_{\min}})$ typed-value queries
& $\Theta(d/R_{\min})$ randomized typed-value queries
& separate opaque-value model; no free values, setup, or loading\\
Full Pauli algebra
& $K$ or $C^TK$
& reversible support enumeration; $\rho_{\mathrm{BE}}=1$
& implicit label rule; $\su(2^n)$ has zero center
& different task model\\
\bottomrule
\end{tabularx}
\end{table}

The conditioning results below govern $\Delta$ and
$\kappa_{\mathrm{spec}}$. The access-dependent block-encoding slack
$\rho_{\mathrm{BE}}$ and any output conversion remain separate resources, as does
an end-to-end comparison with a classical algorithm.

\section{Physical scope of the structural parameters}\label{sec:physical_typicality}

The structural parameters in this section are defined for a specified control set,
its actual real DLA, and a named coefficient metric. Symmetry classes and
Hamiltonian commutants alone leave those data undetermined. The supplied DLA is
the scope of the result; an ensemble model would be an additional assumption.

\paragraph{Conditions for bounded balance.}

Theorem~\ref{thm:beta_structure} yields $\beta=O(1)$ if both $d/r=O(1)$ and $N_{\max}/N_{\min}=O(1)$. A simple algebra has $\beta=1$, including the Gell--Mann frame of Appendix~\ref{app:gell_mann}. These are sufficient conditions. A center-free direct sum can still be unbalanced when its simple blocks have a widening distribution of $N_l$, while $r=o(d)$ produces the separate center-dominated mechanism. Locality, a short symmetry list, or a Cartan label alone proves neither condition.

\paragraph{Specified control families and symmetry objects.}

The Heisenberg and Hubbard Hamiltonians are sums of Pauli strings, but this fact does not identify their Lie closures. The full Pauli algebra has the support enumerator proved in Section~\ref{sec:arithmetic_oracles}. A smaller Pauli-spanned DLA must provide reversible rank/select oracles for its actual row and column supports. The same boundary applies to a symmetry-adapted basis. An isotypic decomposition gives an ambient centralizer, while a chosen control family can generate a proper subalgebra.

\paragraph{Conditioning across compact types.}
Theorem~\ref{thm:obstruction_general} covers every compact simple type under its real-frame hypotheses. The following examples separate a generated algebra from a symmetry or classification label.

Let $\gamma_1,\ldots,\gamma_{2n}$ be Hermitian Clifford generators. The real Lie closure of the full skew-Hermitian bilinear family $\{\gamma_a\gamma_b:a<b\}$ is $\mathfrak{so}(2n)$; the corresponding Gaussian action reaches all of $SO(2n)$~\cite[Theorem~3]{Jozsa2008}. One fixed quadratic Hamiltonian generates only a one-dimensional DLA. The algebraic extension to $2n+1$ Clifford generators gives $\mathfrak{so}(2n+1)$, but it is not, by itself, a boundary-mode control model. For $n\ge3$ both algebras are simple and have $\kappa(A_{\mathrm{gen}})=1$ in any invariant-orthonormal frame. At $n=2$, $\mathfrak{so}(4)\cong\su(2)\oplus\su(2)$ has that value only when the two block normalizations agree. Jordan--Wigner maps the bilinears to Pauli strings. This representation fact supplies neither the typed support enumerators nor their reversible-gate cost, so no implicit-access exponent is assigned here.

The Altland--Zirnbauer and Kitaev classifications organize Hamiltonian constraints and phases by Cartan labels~\cite{AltlandZirnbauer1997,Kitaev2009}. A control set and its generated DLA are additional data. The compact algebras $\uu(n)$, $\mathfrak{usp}(2n)$, and $\mathfrak{so}(2n)$ may be used as algebraic representatives once the full family of symmetry-allowed quadratic controls and its closure have been stated. In the $\uu(n)$ representative, the $\uu(1)$ number direction lies in the center unless it is removed as a global phase or becomes scalar on a fixed-particle sector. These classification labels leave the support-enumeration interface unspecified.

The full real family of centered quadratic bosonic controls spans the metaplectic generators and has real Lie algebra $\mathfrak{sp}(2n,\mathbb{R})$~\cite[Secs.~5--6, Eqs.~(6.3)--(6.4)]{Arvind1995}; a restricted quadratic family can generate a proper subalgebra. The full algebra is noncompact and admits no positive-definite adjoint-invariant inner product, so Lemma~\ref{lem:invariant_reduction} does not apply. Indefiniteness does not merely reverse the sign of Equation~\eqref{eq:CCT_eq_minusK}; a typed support interface for this family would be an additional construction.

The half-filled Hubbard Hamiltonian on a bipartite lattice has commuting spin and $\eta$-pairing symmetries with algebra $\mathfrak{so}(4)\cong\su(2)\oplus\su(2)$~\cite{YangZhang1990}. This identifies a Hamiltonian commutant; the DLA of the standard Hubbard controls is a separate object. If those six symmetry generators are instead supplied as the input algebra and the two simple blocks use the same invariant normalization, then $N_{\mathrm{spin}}=N_{\eta}$ and $\kappa(A_{\mathrm{gen}})=1$. An access cost would additionally require an explicit generator and oracle construction.

For compact inputs, the projector targets the center inside the supplied DLA. A semisimple algebra has zero center, while $\uu(n)$ retains the particle-number direction only under the convention just stated. A charge that commutes with the dynamics without belonging to the generated algebra lies in the larger commutant, a distinct output. Theorem~\ref{thm:zero_center_decision} can test whether this internal center is zero under its additional promises. Exact dimension, a classical basis, and the full commutant remain separate outputs.

\section{Discussion}

The results separate the algebraic projector, the access model used to realize it,
and the downstream task that consumes its output. This separation identifies the
resources supplied by each theorem and those required by a physical application.

\subsection{Scope and Resource Requirements}\label{sec:limitations}

The projector theorem takes a constructed DLA basis, typed access oracles, and a
certified gap as inputs. It returns a block encoding of the coefficient-space projector.
Gap estimation, coefficient-state preparation, and a physical operator interface
are separate stages in an end-to-end application.

The zero-center wrapper assumes a valid maximally entangled preparation, its
inverse, the required reflections, and controlled oracle access. It returns a
bounded-error bit for the promise $R=0$ or $R\geq R_{\min}$. Taking
$R_{\min}=1$ decides whether the center is zero with $O(\sqrt d)$ outer calls;
this can be exponential in the physical qubit count.

Discovery of $\Delta_{\mathrm{cmp}}$, exact rank or basis recovery, and commutant
testing are distinct tasks. The matching lower bounds apply in the
projector-oracle and opaque-value models stated with the results.

\paragraph{Data access: QRAM versus arithmetic oracles.}

An explicit dense tensor requires $O(d^3)$ setup to build both row and column
indices for the typed oracles. Sparse input can reduce this to its indexed input
size. The full Pauli algebra replaces stored indices by the reversible affine
enumerator of Equation~\eqref{eq:pauli_affine_enumerator}; an arbitrary Pauli
subset needs a corresponding rank/select interface.

\paragraph{Ill-conditioned systems and approximate radical computation.}

For an exact radical, the input promise separates zero from every nonzero singular
value. A user may instead specify a numerical threshold $\delta_{th}>0$ and define
the different target
\begin{equation}
\rfrak_{\text{eff}}(\delta_{th}) = \mathrm{span} \{ v_k : \sigma_k(A_{\mathrm{gen}}) < \delta_{th} \}.
\end{equation}
QSVT then uses the supplied normalized threshold
$\delta_{th}/\alpha$. The threshold is part of the numerical target specified by
the user. Singular-value estimation reports components supported by its input
state; locating the global smallest nonzero singular value additionally requires
rank and overlap promises.

\subsection{Relation to Classical Simulability}\label{sec:simulability_significance}

A recent line of work locates the relevant classical-simulability regime. When the DLA has dimension polynomial in the number of qubits, Lie-algebraic (g-sim) methods simulate the Heisenberg evolution of suitable observables in classical polynomial time~\cite{Goh2025gsim}. Building on such constructions, Cerezo and coworkers argue, case by case, that structural features associated with the absence of barren plateaus often also supply a classical surrogate, making provable trainability and quantum advantage difficult to secure together~\cite{Cerezo2025dichotomy}. Their perspective presumes an initial data-acquisition phase on quantum devices and concerns observable simulation and trainability.

The present procedure returns a coefficient-space radical projector, whereas g-sim
addresses the Heisenberg evolution of selected observables. An end-to-end
comparison would therefore require a common input model and output task. The
conditioning law of Theorem~\ref{thm:obstruction_general} is an algebraic statement
independent of that comparison.

\section{Applications}
For compact DLAs the projector targets the center generated within the algebra. A
quantity that commutes with the dynamics but is not generated by it, such as
$\eta$-pairing in a sparse local model~\cite{Yang1989}, lies in the commutant outside
this center~\cite{MoudgalyaMotrunich2022}. Enlarging the generating set changes the
input algebra and its access interface.

\subsection{Coefficient-Space Use}\label{sec:hybrid_full_decomposition}

Let $U_{P_{\rfrak}}$ be the block encoding returned by the theorem and suppose that
a normalized coefficient state $\ket{h}$ is supplied by a separate preparation
routine. The signal block satisfies
\begin{equation}\label{eq:projector_decomposition}
\left\|
(\bra{0}^{\otimes a_P}\otimes I)U_{P_{\rfrak}}
(\ket{0}^{\otimes a_P}\otimes\ket{h})
-P_{\rfrak}\ket{h}
\right\|\leq\epsilon.
\end{equation}
Writing
$\widetilde P_{\rfrak}=(\bra{0}^{\otimes a_P}\otimes I)
U_{P_{\rfrak}}(\ket{0}^{\otimes a_P}\otimes I)$, the actual postselection
probability is $p_h=\|\widetilde P_{\rfrak}\ket{h}\|^2$. Both
$\widetilde P_{\rfrak}$ and $P_{\rfrak}$ are contractions, so
\[
\bigl|p_h-\|P_{\rfrak}\ket{h}\|^2\bigr|
\leq
(\|\widetilde P_{\rfrak}\|+\|P_{\rfrak}\|)
\|\widetilde P_{\rfrak}-P_{\rfrak}\|
\leq2\epsilon.
\]
The projector block encoding can be combined with a separately supplied preparation
routine and overlap lower bound. Dimension estimation and classical-basis recovery
are additional output tasks. Theorem~\ref{thm:zero_center_decision} assumes a
dedicated preparation of $\ket{\Phi_d}$ and its inverse; a general preparation
routine for $\ket h$ is a separate interface.

Exact classical methods over $\mathbb Q$ can continue from an explicit radical
basis to a Levi complement, and a semisimple algebra can be decomposed into simple
ideals under the source-specific assumptions of the cited algorithms
~\cite{deGraaf1997Levi,deGraaf1997Semisimple}. These methods begin from an explicit
radical basis and use a distinct cost model.

\subsection{Algebraic Splitting of Compact Dynamics}\label{sec:trotter_error_reduction}

For a compact DLA, let $h_z=P_{\zfrak}h$ and
$h_s=(I-P_{\zfrak})h$ in the chosen coefficient basis, and define
$H_z=\sum_j(h_z)_jB_j$ and $H_s=\sum_j(h_s)_jB_j$. Centrality gives
$[H_z,H_s]=0$, hence
\begin{equation}\label{eq:trotter_free}
e^{-itH}=e^{-itH_s}e^{-itH_z}.
\end{equation}
The central evolution need not be trivial.

\begin{remark}[Scope of the algebraic split]\label{rmk:small_center_value}
Equation~\eqref{eq:trotter_free} is algebraic. The projector acts on coefficient
space, while a block encoding of $H$ acts on the physical Hilbert space, so a
sandwich of those unitaries is not type-correct. Implementing $H_z$ or $H_s$
requires separate coefficient PREPARE maps, a SELECT oracle for the basis
operators, an LCU normalization, error propagation, and a simulation cost. Those
interfaces are outside the present algorithmic contract.
\end{remark}

\section{Conclusion}

The construction projects coefficient states onto the radical in two input
regimes. A general real or complex Lie algebra uses $C^TK$, while a compact DLA
uses $K$ itself. Typed row-support, column-support, and value oracles encode a
fixed $b$-bit structure tensor; a certified nonzero-singular-value gap and an
operator-norm tolerance set the filter. The resulting output is an approximate
block encoding of the projector on coefficient space.

For the compact target, normalized projector-rank estimation also yields a
zero-center decision. Under the promise $R=0$ or $R\geq R_{\min}$, the test uses
$O(\sqrt{d/R_{\min}}\log(1/\delta_{\mathrm{fail}}))$ outer projector calls and
returns a bounded-error bit. The postselected exact signal is a normalized
purification of the center projector. Exact rank retains worst-case linear outer
complexity in the fixed-error projector-counting model.

The query factor separates into $\alpha/\Delta$, equivalently
$\rho_{\mathrm{BE}}\kappa_{\mathrm{spec}}$ for a tight supplied gap. The full
Pauli algebra admits reversible support enumeration with
$\rho_{\mathrm{BE}}=1$; a smaller Pauli-spanned family needs its own rank/select
access. After sparse Killing-form construction, the general and compact classical
calculations have respective $O(d^4)$ and $O(d^3)$ rank costs. These measures
describe different outputs and resource units.

Theorem~\ref{thm:opaque_value_separation} complements the explicit-input results
with a compact-Lie family that embeds promised search while preserving the public
structure-constant support and certified metadata. In the opaque-value model, the
family realizes the standard quadratic quantum-versus-randomized-classical query
separation. The Lie-algebraic contribution is the promise-preserving construction;
the search bounds determine the query exponents. Readable tables, setup, loading,
projector synthesis, and physical gates belong to different resource models.

The conditioning law identifies the intrinsic spectral contribution to block-encoded
radical projection. Together with the compact center test, it turns the structural
decomposition of a supplied Lie algebra into a controlled quantum operation and a
falsifiable decision task. Classical-basis recovery and physical simulation can
then be studied as explicit downstream interfaces.

\appendix

\section{Normalization factor derivations}\label{sec:derivation_normalization_factor}

The normalization proof uses both row and column access. Let
$M\in\mathbb F^{R\times C}$ have at most $r_M$ nonzeros in each row, at most
$c_M$ in each column, and entry magnitude at most $F$. Then
\begin{equation}
\|M\|_2
\leq\sqrt{\|M\|_1\|M\|_\infty}
\leq F\sqrt{r_Mc_M}.
\label{eq:sparse_norm_proof}
\end{equation}
Lemma~48 of Gilyén et al.\ realizes this upper bound as the subnormalization when
the row enumerator, column enumerator, and value oracle are supplied
~\cite{Gilyen2019}. A maximum column norm alone is not sufficient.

For $B$, a fixed row $(k,m)$ contains at most $s_{int}$ nonzeros and a fixed
column $j$ at most $s_{adj}$. The transpose factor has the reversed bounds.
For $C^T$, Equation~\eqref{eq:commutator_sparsities} gives the two required
bounds. Substitution in Equation~\eqref{eq:sparse_norm_proof} gives the
constituent scales stated in the main text and hence
Equation~\eqref{eq:target_normalizations}. These implemented scales may
exceed the target spectral norm.

If the constituent block errors are $\eta_B$, $\eta_A$, and $\eta_C$, multiplication
gives
\begin{align}
\eta_{\mathrm{cmp}}
&\leq\alpha_{A^T}\eta_B+\alpha_B\eta_A,\nonumber\\
\eta_{\mathrm{gen}}
&\leq\alpha_{C^T}\alpha_{A^T}\eta_B
 +\alpha_{C^T}\alpha_B\eta_A
 +\alpha_{A^T}\alpha_B\eta_C.
\label{eq:composite_block_errors}
\end{align}
For either target, let $\alpha$ be its composite normalization and let
$D=O((\alpha/\Delta)\log(1/\epsilon))$. The constituent rotations are synthesized
so that the corresponding left-hand side of
Equation~\eqref{eq:composite_block_errors} is at most
$\alpha\epsilon^2/(256D^2)$. The same stability bound then contributes
at most $4D\sqrt{\eta_{\mathrm{tar}}/\alpha}\leq\epsilon/4$. An ideal-filter
error of at most $\epsilon/2$ and phase-synthesis error of at most $\epsilon/4$
complete the stated operator-norm error budget.

For the full Pauli algebra, the affine enumerator gives
$q=2^{2n-1}$ support positions per nonzero label and coefficient magnitude two.
The deterministic product label makes the relevant supports disjoint, so direct
Gram calculations give
\begin{equation}
B^\dagger B=4qI,\qquad
A^T(A^T)^\dagger=4qI,\qquad
(C^T)^\dagger C^T=4qI,\qquad
K=-4qI.
\label{eq:full_pauli_gram_guard}
\end{equation}
Thus each constituent scale is $2\sqrt q$, the compact product scale is $4q$,
and the general product scale is $8q^{3/2}$. The exact full-Pauli Gram calculation
in Equation~\eqref{eq:full_pauli_gram_guard} proves $\rho_{\mathrm{BE}}=1$ for this
named family.

\section{Zero-center proofs and an opaque-value family}
\label{app:zero_center_details}

\subsection{Decision gap, preparation error, and the center state}

The valid-label predicate in Section~\ref{sec:zero_center_decision} is part of
the good reflection. Measuring only the signal ancilla is insufficient. A
unitary may map a valid input into a padded output label while its
valid-to-valid block is zero, which would create a false positive if the padded
label were accepted.

For the decision estimate, let $a_0=R_{\min}/d$ and $x=\sqrt{a_0}$. Under the error
choice in Theorem~\ref{thm:zero_center_decision}, the two success-probability
intervals imply
\[
\theta_0\leq\arcsin\frac{x}{4},
\qquad
\theta_1\geq\arcsin\frac{3x}{4},
\qquad
\theta=\arcsin\sqrt q.
\]
The claimed angle separation follows from
\begin{equation}
\arcsin\frac{3x}{4}-\arcsin\frac{x}{4}\geq\frac{x}{2},
\qquad 0\leq x\leq1.
\label{eq:zero_center_angle_gap}
\end{equation}
To verify it, subtract $x/2$ from the left-hand side and differentiate. With
$u=x^2/16$, the derivative is
\[
\frac14\left[
3(1-9u)^{-1/2}-(1-u)^{-1/2}
\right]-\frac12.
\]
The expression in square brackets equals two at $u=0$ and has derivative
\[
\frac{27}{2}(1-9u)^{-3/2}
-\frac12(1-u)^{-3/2}>0
\]
on $0\leq u\leq1/16$. The difference in
Equation~\eqref{eq:zero_center_angle_gap} is therefore nondecreasing from zero.

\begin{proof}[Proof of Corollary~\ref{cor:zero_center_unitary_prep}]
Put $\widetilde A=\|(B\otimes I)\ket{\widetilde\Phi}\|$ and retain all work
registers in $\ket{\widetilde\Phi}$. Under $H_0$, Equation
\eqref{eq:valid_center_block} gives $\|B\|\leq\eta$, so
\[
\widetilde A\leq\frac{x}{4}.
\]
Under $H_1$, contraction and the preparation error give
\[
\widetilde A
\geq x(1-\eta)-\xi
\geq\frac{5x}{8}.
\]
The two amplitude intervals have a constant multiple of the original gap.
The actual preparation circuit supplies the reflection through its output
state because its exact inverse is available.
\end{proof}

The pure-state condition in this corollary is material. If only a reduced
density matrix within trace distance $\zeta$ of $I_d/d$ is supplied, a binary
probability may shift by $\zeta$. A sufficient probability scale is then
$\zeta=O(a_0)$, together with a state-preparation unitary for a purification, its inverse, and the
corresponding reflection. The theorem does not cover unknown garbage or an
uncontrolled noisy channel.

\begin{proof}[Proof of Proposition~\ref{prop:center_state_output}]
Vectorization of the exact good branch gives
\[
(P\otimes I)\ket{\Phi_d}
=\frac{\operatorname{vec}(P)}{\sqrt d}.
\]
Its squared norm is $R/d$. Normalization proves the stated pure state, and the
two partial traces are $PP^\dagger/R=P/R$ and
$(P^\dagger P)^T/R=P^T/R$.

Write $T(\rho,\sigma)=\tfrac12\|\rho-\sigma\|_1$, and abbreviate
$T(\ket\psi,\ket\phi)=T(\ket\psi\!\bra\psi,\ket\phi\!\bra\phi)$ for pure states.
Put $p=R/d$ and $y_B=(B\otimes I)\ket{\Phi_d}$. If $\eta<\sqrt p$, then
$\|y_B\|\geq\sqrt p-\eta>0$, so $t=\Tr(B^\dagger B)=d\|y_B\|^2>0$.
The normalized good state and its reduced states are
\begin{equation}
\ket{\Psi_B}=\frac{\operatorname{vec}(B)}{\sqrt t},
\qquad
\rho_A^B=\frac{BB^\dagger}{t},
\qquad
\rho_R^B=\frac{(B^\dagger B)^T}{t}.
\label{eq:approximate_center_state}
\end{equation}
The unnormalized exact branch has norm $\sqrt p$. Its distance from $y_B$ is at
most $\eta$. Normalization gives
\[
T(\Psi_B,\Psi_P)
\leq\min\left\{1,\frac{2\eta}{\sqrt p}\right\}.
\]
Partial trace cannot increase this distance. For preparation error
$\xi$ from the stated unitary circuit, put $\widetilde y=(B\otimes I)\ket{\widetilde\Phi}$. If
$\eta+\xi<\sqrt p$, then $\|\widetilde y\|\geq\sqrt p-(\eta+\xi)>0$; define
$\ket{\widetilde\Psi_B}=\widetilde y/\|\widetilde y\|$. The same argument gives
\[
T(\widetilde\Psi_B,\Psi_P)
\leq\min\left\{1,\frac{2(\eta+\xi)}{\sqrt p}\right\}.
\]
Thus a requested output-state error $\zeta_{\mathrm{out}}$ calls for
$\eta+\xi=O(\zeta_{\mathrm{out}}\sqrt{R_{\min}/d})$.
\end{proof}

When the null branch is exactly noiseless and a reliable lower bound
$\lambda>0$ on the target overlap is known, fixed-point search offers the
upper bound
$O(\lambda^{-1/2}\log(1/\delta_{\mathrm{fail}}))$
~\cite[Eq.~(3)]{YoderLowChuang2014}. An approximate projector permits nonzero
null-branch leakage. That fixed-point statement alone does not distinguish the
two intervals used in Theorem~\ref{thm:zero_center_decision}.

\subsection{Rank estimation and the projector-oracle boundary}

For an exact projector, amplitude estimation of $q=R/d$ to additive error
$\varepsilon_p$ takes
\[
O\!\left(\varepsilon_p^{-1}\log(1/\delta_{\mathrm{fail}})\right)
\]
outer calls. The exact integer task behaves differently. The exact-counting
algorithm has
input-dependent expected cost
\[
\Theta\!\left(\sqrt{(R+1)(d-R+1)}\right)
\]
and success probability at least $2/3$
~\cite[Theorem~17]{brassard2000quantum}. This cost is $O(\sqrt d)$ at the two
endpoints and $O(d)$ near the middle. The lower bound of
Nayak and Wu applied to adjacent middle weights is $\Omega(d)$
~\cite[Corollary~1.2]{NayakWu1999}. The fixed-error worst case for exact rank is
therefore $\Theta(d)$.

Lemma~\ref{lem:center_rank_overlap} also gives
\[
\left|q-\frac Rd\right|\leq2\eta.
\]
Exact recovery for every $R$ requires normalized error below half a rank step,
so a safe worst-case choice is $\eta=O(1/d)$. Combining this precision with
Corollary~\ref{cor:compact_optimization} gives the upper bound
\begin{equation}
O\!\left[
\frac{F^2s_{adj}s_{int}}{\Delta_{\mathrm{cmp}}}
\;d\log d\log\frac1{\delta_{\mathrm{fail}}}
\right]
\label{eq:exact_center_rank_typed_upper}
\end{equation}
for typed structure-oracle calls, in addition to state-preparation and gate
costs.

\begin{theorem}[Projector-oracle decision lower bound]
\label{thm:projector_oracle_zero_center_lower_bound}
In a canonical Boolean membership oracle, or a diagonal-projector dilation that
reveals no extra information, fixed-error distinction of rank zero from rank
$R_{\min}$ requires
\[
\Omega\!\left(\sqrt{\frac d{R_{\min}}}\right)
\]
quantum queries.
\end{theorem}

\begin{proof}
Take $P_x=\operatorname{diag}(x_1,\ldots,x_d)$ and promise that the Hamming
weight of $x$ is zero or $R_{\min}$. A zero-center decision distinguishes these
two symmetric partial-function values. Corollary~1.2 of Nayak and Wu gives the
stated bound~\cite{NayakWu1999}.
\end{proof}

At constant error, this theorem determines the outer exponent in the no-leak
projector model. The confidence logarithm, explicit-Lie-input costs, and projector
synthesis are separate questions.

\subsection{The common-support compact Lie construction}
\label{app:opaque_value_family}

Start with the adapted six-dimensional algebras
\[
\mathfrak h_0=\su(2)\oplus\su(2),
\qquad
\mathfrak h_1=\su(2)\oplus\uu(1)^{\oplus3}.
\]
Each $\su(2)$ triple obeys $[e_a,e_b]=\epsilon_{abc}e_c$. Use the same
integer change of basis for both algebras,
\begin{equation}
B_i=\sum_{a=1}^{6}S_{ai}e_a,\qquad
S=
\begin{pmatrix}
-1&-3& 1& 2& 4& 3\\
 2& 1& 4& 1& 1& 0\\
 3& 3& 2&-2&-3&-3\\
 2& 2& 1&-1&-2&-2\\
 1& 0& 2& 0& 1& 0\\
 1& 1& 1& 0&-1&-1
\end{pmatrix},
\qquad \det S=1.
\label{eq:same_support_basis_change}
\end{equation}
If $f^{(\chi)}$ is the tensor in the adapted basis, the transformed constants
are
\begin{equation}
g^{(\chi)}_{ijk}
=\sum_{a,b,c}
S_{ai}S_{bj}f^{(\chi)}_{abc}(S^{-1})_{kc},
\qquad \chi\in\{0,1\}.
\label{eq:same_support_tensor_transform}
\end{equation}

\begin{lemma}[Certificates for the two blocks]
\label{lem:same_support_compact_blocks}
The tensors in Equation~\eqref{eq:same_support_tensor_transform} have the
following exact properties.
\begin{enumerate}
    \item Both satisfy Jacobi and
    $g^{(\chi)}_{ijk}\ne0$ exactly when $i\ne j$. Their common ordered
    nonzero count is $180$.
    \item Their center dimensions are zero and three, respectively.
    \item The maximum entry magnitudes are $37$ and $44$. The common padded
    sparsities are
    \[
    (s_{adj},s_{int},s_{\mathrm{br}},s_{\mathrm{out}})
    =(30,5,6,30).
    \]
    \item One fixed one-based table entry obeys
    $g^{(0)}_{121}=7$ and $g^{(1)}_{121}=3$.
    \item The two Killing matrices have nonzero singular-value gaps above
    $1/32$.
\end{enumerate}
\end{lemma}

\begin{proof}
Integer substitution in Equation~\eqref{eq:same_support_tensor_transform}
gives the support, Jacobi identities, entry bounds, and the fixed entry. The
centers follow from the two adapted direct sums because the basis change is
invertible.

For the gap certificate, let $A$ be the first three rows of $S$. The Killing
matrices are
\[
K_0=-2S^TS,\qquad K_1=-2A^TA.
\]
Exact elimination gives
\[
\|S^{-1}\|_F^2=58,\qquad
\det(AA^T)=1245,\qquad
\Tr(AA^T)=107.
\]
It follows that
\begin{equation}
\sigma_{\min}(K_0)\geq\frac1{29},
\qquad
\sigma_{\min}^{+}(K_1)\geq\frac{2490}{11449}.
\label{eq:same_support_gap_certificates}
\end{equation}
Both lower bounds exceed $1/32$. Direct counting of the common support gives
the four padded sparsities.
\end{proof}

\begin{proof}[Detailed query reduction for
Theorem~\ref{thm:opaque_value_separation}]
Let the hidden bit $x_j$ choose the block in position $j$. All support oracles
are independent of $x$. A typed-value query can compute its block index,
query $x_j$, select one of the two fixed $90$-entry integer tables, and
uncompute the work registers with constant Boolean-query overhead. Conversely,
one query to entry $(1,2,1)$ of block $j$ distinguishes the values seven and
three, so it returns $x_j$ with constant typed-query overhead. The two oracle
models are query-equivalent up to constants.

Under the promise $|x|=0$ or $|x|=t$, the lower bound
$\Omega(\sqrt{N/t})$ follows from Corollary~1.2 of
Nayak and Wu~\cite{NayakWu1999}. Direct Grover search on the fixed entry gives
$O(\sqrt{N/t})$, proving the quantum statement in
Equation~\eqref{eq:opaque_value_zero_center_separation}.

For a classical algorithm, uniformly sampling $O(N/t)$ distinct blocks gives
the upper bound. For the lower bound, use a distribution that chooses the
all-zero input with probability one half and a uniform $t$-subset with
probability one half. After $k$ distinct queries, the no-hit probability on
the second branch is
\[
\frac{\binom{N-t}{k}}{\binom Nk}.
\]
It remains bounded below by a positive constant when
$k\leq cN/t$ for a sufficiently small constant $c$. On a no-hit transcript,
the two branches are identical. Yao's principle therefore gives
$\Omega(N/t)$ randomized queries.
\end{proof}

The construction has only $O(d)$ local numerical data. If that table is
explicitly supplied or if setup may scan it without charge, a classical
routine can read the fixed entry in each block and the query separation no
longer describes the full cost. Direct Grover search is already optimal for
this family, and a block-index-controlled projector block encoding can be implemented with a
constant number of hidden-bit queries. These facts keep the decision result
separate from QSVT necessity and projector-synthesis complexity.

\section{Pauli-basis structural proofs}\label{app:pauli_proofs}

In the Pauli basis, the commutator Gram matrix is diagonal and the conditioning analysis of Section~\ref{sec:condition_number_barrier} takes a closed form. This specialization follows from the invariant-orthonormal real-frame reduction of Appendix~\ref{app:invariant_reduction}.

\begin{lemma}[Diagonality of $CC^T$ in the Pauli Basis]\label{lemma:CCT_diagonal}
For any DLA $\g$ represented in the standard Skew-Hermitian Pauli basis (where $f_{ijk}$ are real), the Gram matrix $CC^T$ is diagonal: $C C^T = \text{diag}(N_k)$.
\end{lemma}
\begin{proof}
The elements of the Gram matrix are $(CC^T)_{k,k'} = \sum_{i,j} f_{ijk} f_{ijk'}$.

If this term is non-zero, there must exist a pair $(i,j)$ such that $f_{ijk} \neq 0$ AND $f_{ijk'} \neq 0$. This implies $[B_i, B_j]$ has components along both $B_k$ and $B_{k'}$.

The Pauli basis possesses a structure inherited from the underlying Clifford algebra and the associated finite Pauli group~\cite{Gottesman1998}. The commutator of two Pauli operators $[B_i, B_j]$ is proportional to a unique third Pauli operator, or is zero. Due to this algebraic determinism, the result of the commutator is unique. Therefore, we must have $k=k'$.

The diagonal elements are $(CC^T)_{k,k} = \sum_{i,j} f_{ijk}^2 = N_k$.
\end{proof}

In the Pauli basis, this diagonality reduces the condition-number analysis to the weighted Killing form.

\begin{theorem}[Simplification of $\kappa(A_{\mathrm{gen}})$ in the Pauli Basis]\label{thm:kappa_simplification}
Let $D = \text{diag}(\sqrt{N_k})$ be the diagonal matrix of the square roots of the row norms. If the DLA $\g$ is represented in the standard Skew-Hermitian Pauli basis and has nonzero semisimple part $\mathfrak{s}$, the relevant condition number of $A_{\mathrm{gen}}$ is determined by the operator restricted to $\mathfrak{s}$. Specifically:
\begin{equation}
\kappa(A_{\mathrm{gen}}) = \kappa((DK)|_{\mathfrak{s}}).
\end{equation}
\end{theorem}
\begin{proof}
We analyze the singular values of $A_{\mathrm{gen}}$ via its Gram matrix $G$. Since the Pauli basis guarantees $f_{ijk} \in \mathbb{R}$, the Gram matrix is $G = A_{\mathrm{gen}}^T A_{\mathrm{gen}} = K C C^T K$. By Lemma~\ref{lemma:CCT_diagonal}, $C C^T = \text{diag}(N_k) = D^2$. Thus, $G = K D^2 K = (DK)^T (DK)$.
For central elements $B_k \in \zfrak$, $N_k=0$, implying $D_{kk}=0$. Hence $DK$ vanishes on the center. The non-zero singular values of $A_{\mathrm{gen}}$ are therefore the singular values of $DK$ restricted to the semisimple component $\s$, where $D$ is strictly positive definite.
\end{proof}

\begin{proof}[Proof of Theorem~\ref{thm:well_conditioning_su2n}]

Theorem~\ref{thm:kappa_simplification} reduces the claim to $\kappa(DK)$. Analyze $K$ and $D$ separately.

\textbf{Killing form $K$.} Here $\g=\su(N)$ with $N=2^n$ is a simple Lie algebra, so the Killing form is proportional to the Trace form. The Pauli basis is orthonormal under the Trace form, so $K\propto I$ and $\kappa(K)=O(1)$.

\textbf{Norm matrix $D=\text{diag}(\sqrt{N_k})$.} The squared norm $N_k=\sum_{i,j}|f_{ijk}|^2$ counts the structure-constant weight that lands on $B_k$ from commutators of basis pairs. The Pauli basis inherits the uniform structure of the Clifford algebra and its associated finite group~\cite{Gottesman1998, Appleby2005}, so $N_k$ is independent of $k$ provided $B_k\ne iI$, which holds on the basis of $\su(2^n)$. Hence $N_k=\overline N$ for all $k$, $D=\sqrt{\overline N}\,I$, and $\kappa(D)=O(1)$.

Both factors are scalar, so $DK\propto I$ and $\kappa(A_{\mathrm{gen}})=O(1)$.
\end{proof}

\begin{proof}[Proof of Corollary~\ref{cor:simple_summand_obstruction} (direct Pauli count)]

Fix a Pauli basis of $\g$ adapted to the direct-sum decomposition $\g=\s_1\oplus\cdots\oplus\s_m\oplus\zfrak$. The hypothesis that each $\s_l$ is itself Pauli-spanned makes such a refinement possible. Distinct simple ideals commute, $[\s_l,\s_{l'}]=0$ for $l\ne l'$, and $[\zfrak,\g]=0$. The defining relation $[B_i,B_j]=\sum_k f_{ijk}\,B_k$ forces $f_{ijk}=0$ whenever the indices $i,j,k$ do not all lie in a single simple summand. For $k\in\s_l$,
\begin{equation}\label{eq:Nk_internal_app}
N_k \;=\; \sum_{i,j} f_{ijk}^{\,2} \;=\; \sum_{i,j\in\s_l} f_{ijk}^{\,2}.
\end{equation}
Only pairs internal to $\s_l$ contribute.

The skew-Hermitian Pauli basis is orthonormal for the inner product $\langle X,Y\rangle=-2^{-n}\,\mathrm{tr}_{\mathbb{C}^{2^n}}(XY)$, since $\mathrm{tr}_{\mathbb{C}^{2^n}}(P_jP_k)=2^{n}\delta_{jk}$, and the form is invariant by cyclicity of the trace. Adapted as above to the splitting, the basis meets every hypothesis of Corollary~\ref{cor:invariant_block_ratio}: on each simple summand the Killing matrix is block-scalar, $K|_{\s_l}=-N_l\,I_{d_l}$ with $N_l>0$. Theorem~\ref{thm:kappa_poly_bound} keeps the whole matrix strictly diagonal. The identity \eqref{eq:CCT_eq_minusK} evaluates the diagonal as $K_{kk}=-(CC^{T})_{kk}=-N_{k}$, so $N_k=N_l$ for every $k\in\s_l$.

A direct combinatorial count fixes the numerical value. By hypothesis the Pauli basis of $\s_l$ corresponds, via symplectic indices, to all non-zero vectors of a $2n_l$-dimensional non-degenerate symplectic subspace of $\mathbb{F}_2^{2n}$, so the non-identity Paulis spanning $\s_l$ number exactly $d_l=4^{n_l}-1$, and two basis elements anti-commute iff their symplectic inner product equals $1$. Fix $B_i\in\s_l$. Among the $d_l-1$ other non-identity Paulis of $\s_l$, the number with symplectic inner product $1$ against $B_i$ equals $2^{2n_l-1}=(d_l+1)/2$; the number with inner product $0$ equals $2^{2n_l-1}-1=(d_l-1)/2$, where the off-by-one reflects the fact that $B_i$ itself trivially commutes with $B_i$. The count of anti-commuting partners is therefore $(d_l+1)/2$. Each anti-commuting ordered pair $(B_i,B_j)$ produces a unique $B_k$ with $|f_{ijk}|=2$ by the Pauli product law and the normalization $f_{ijk}\in\{0,\pm 2\}$. Summing,
\[
\sum_{i,j,k\in\s_l} f_{ijk}^{\,2} \;=\; d_l\cdot\frac{d_l+1}{2}\cdot 4 \;=\; 2\,d_l(d_l+1).
\]
This equals $\sum_{k\in\s_l} N_k=d_l\cdot N_k$ by the constancy of $N_k$ on $\s_l$ just established. Division yields
\begin{equation}\label{eq:Nk_explicit_app}
N_k \;=\; 2(d_l+1) \qquad \text{for every } k\in\s_l.
\end{equation}
Substituting into the block-scalar form gives $K|_{\s_l}=-2(d_l+1)\,I_{d_l}$.

Substituting $N_l=2(d_l+1)$ into Theorem~\ref{thm:obstruction_general} gives $\kappa(A_{\mathrm{gen}})=\bigl((d_{\max}+1)/(d_{\min}+1)\bigr)^{3/2}$, with $\kappa(D|_{\s})$ the square root of the same ratio; the saturating case $d_{\min}=O(1)$, $d_{\max}=\Theta(d)$, realised by $\s=\su(2)\oplus\su(2^k)$, attains the Pauli ceiling $\kappa(A_{\mathrm{gen}})=\Theta(d^{3/2})$.
\end{proof}

\section{Invariant-orthonormal real-frame reduction of the conditioning}\label{app:invariant_reduction}

Lemma~\ref{lemma:CCT_diagonal} derives the diagonality of $CC^{T}$ from the Pauli product law, and Theorem~\ref{thm:kappa_simplification} then trades $\kappa(A_{\mathrm{gen}})$ for $\kappa((DK)|_{\s})$. The diagonal form is incidental. The reduction below uses one identity between the commutator Gram matrix and the Killing matrix in a real basis orthonormal for a positive invariant inner product and adapted to $\g=\s\oplus\zfrak$.

\begin{lemma}[Invariant reduction of $\kappa(A_{\mathrm{gen}})$]\label{lem:invariant_reduction}
Let $\g$ be a compact real Lie algebra with $\g=\s\oplus\zfrak$, where $\s=[\g,\g]$ is semisimple and $\zfrak$ is the center. Let $\langle\cdot,\cdot\rangle$ be a positive real invariant inner product, so that $\langle[Z,X],Y\rangle=-\langle X,[Z,Y]\rangle$ for all $X,Y,Z\in\g$, and let $\{B_k\}_{k=1}^{d}$ be a real coefficient basis orthonormal for it and adapted to $\g=\s\oplus\zfrak$. Then the real structure constants are totally antisymmetric, and the real commutator matrix $C$ and Killing matrix $K$ satisfy
\begin{equation}\label{eq:CCT_eq_minusK}
CC^{T}=-K,
\end{equation}
while the Euclidean Gram matrix of $A_{\mathrm{gen}}=C^{T}K$ is $G=A_{\mathrm{gen}}^{T}A_{\mathrm{gen}}=-K^{3}$. If $\s\ne0$, the nonzero singular values of $A_{\mathrm{gen}}$ are $|\mu|^{3/2}$ as $\mu$ ranges over the eigenvalues of $K|_{\s}$, so
\begin{equation}\label{eq:kappa_three_halves}
\kappa(A_{\mathrm{gen}})=\kappa(K|_{\s})^{3/2}.
\end{equation}
If $\s=0$, then $C=K=A_{\mathrm{gen}}=0$ and this nonzero-spectrum condition number is undefined.
\end{lemma}

\begin{unnumberedremark}[Coefficient-frame scope]
The lemma permits complex skew-Hermitian matrix representations of $\g$,
provided the coefficient algebra and basis are real. It does not extend to an
arbitrary complex coefficient frame by replacing the Euclidean transpose with
the Hermitian adjoint.
\end{unnumberedremark}

\begin{proof}
Orthonormality identifies the structure constants with inner products. From $[B_i,B_j]=\sum_k f_{ijk}B_k$ and $\langle B_k,B_{k'}\rangle=\delta_{kk'}$,
\[
f_{ijk}=\langle[B_i,B_j],B_k\rangle.
\]
Antisymmetry under $i\leftrightarrow j$ is the antisymmetry of the bracket. The remaining transposition follows from invariance with $Z=B_i$, which reads $\langle[B_i,B_j],B_k\rangle=-\langle B_j,[B_i,B_k]\rangle=-f_{ikj}$, hence $f_{ijk}=-f_{ikj}$. Adjacent transpositions generate the symmetric group on three letters, so $f$ is totally antisymmetric.

That symmetry is the entire mechanism. Write the commutator Gram matrix entrywise, then move the third index to the front through the cyclic shift $f_{ijk}=f_{kij}$, an even permutation that carries no sign:
\[
(CC^{T})_{kk'}=\sum_{i,j}f_{ijk}\,f_{ijk'}=\sum_{i,j}f_{kij}\,f_{k'ij}.
\]
The Killing matrix of Section~\ref{sec:killing_form_factorization} is $K_{kk'}=\sum_{i,j}f_{kij}f_{k'ji}$. Reverse the last two indices of the second factor, $f_{k'ji}=-f_{k'ij}$, and the two sums agree up to a sign:
\[
K_{kk'}=-\sum_{i,j}f_{kij}\,f_{k'ij}=-(CC^{T})_{kk'}.
\]
This proves \eqref{eq:CCT_eq_minusK}. On the center the identity is the trivial $0=0$, since a central $B_k$ is never the image of a bracket, forcing $f_{ijk}=0$ there, and $K$ annihilates $\zfrak$.

The reduction is now pure algebra. With real matrices and $K^{T}=K$,
\[
G=A_{\mathrm{gen}}^{T}A_{\mathrm{gen}}=(C^{T}K)^{T}(C^{T}K)=K\,(CC^{T})\,K=-K^{3}.
\]
When $\s\ne0$, diagonalize the real symmetric matrix $K|_{\s}$ by the spectral theorem. Its eigenvalues $\mu$ are real and strictly negative, because the Killing form of a compact semisimple algebra is negative definite. The matching eigenvalue of $G$ is $-\mu^{3}=|\mu|^{3}>0$, so $G|_{\s}$ is positive definite and the singular values of $A_{\mathrm{gen}}$ on $\s$ are $|\mu|^{3/2}$. The center contributes only zeros. Dividing the largest by the smallest,
\[
\kappa(A_{\mathrm{gen}})=\frac{\max_{\mu}|\mu|^{3/2}}{\min_{\mu}|\mu|^{3/2}}=\Bigl(\frac{\max_{\mu}|\mu|}{\min_{\mu}|\mu|}\Bigr)^{3/2}=\kappa(K|_{\s})^{3/2},
\]
the last step using that the singular values of the symmetric matrix $K|_{\s}$ are the moduli $|\mu|$.
\end{proof}

The Pauli statement of Section~\ref{sec:condition_number_barrier} is one corner of this lemma. In the Pauli basis $K$ is diagonal by Theorem~\ref{thm:kappa_poly_bound}, so \eqref{eq:CCT_eq_minusK} forces $CC^{T}$ diagonal as well, which recovers Lemma~\ref{lemma:CCT_diagonal} with $N_k=-K_{kk}$; the substitution $D^{2}=CC^{T}=-K$ then collapses $\kappa((DK)|_{\s})$ to the same $\kappa(K|_{\s})^{3/2}$. Neither integrality nor the Clifford group entered the argument.

\begin{corollary}[Block-scalar form and the dimensional ratio]\label{cor:invariant_block_ratio}
Suppose further that $\s=\bigoplus_{l=1}^{m}\s_l$ with $m\ge1$, each $\s_l$ simple, and that the basis is adapted to $\bigoplus_l\s_l$. Then on the orthonormal basis $K|_{\s_l}=-N_l\,I_{d_l}$ for a positive scalar $N_l$, and
\[
\kappa(A_{\mathrm{gen}})=\Bigl(\frac{\max_l N_l}{\min_l N_l}\Bigr)^{3/2}.
\]
\end{corollary}

\begin{proof}
Pick $X,Y\in\s_l$. The adjoint $\ad_X$ respects the decomposition and acts as zero on every $\s_{l'}$ with $l'\ne l$, since $[X,\s_{l'}]\subseteq[\s_l,\s_{l'}]=0$. On the center it also vanishes. The same is true of $\ad_Y$. The composition $\ad_X\ad_Y$ therefore vanishes outside $\s_l$, and the trace decomposes:
\begin{equation}\label{eq:K_restriction_app}
K_\g(X,Y) \;=\; \mathrm{tr}_{\g}\!\bigl(\ad_X\,\ad_Y\bigr) \;=\; \mathrm{tr}_{\s_l}\!\bigl(\ad_X|_{\s_l}\,\ad_Y|_{\s_l}\bigr) \;=\; K_{\s_l}(X,Y).
\end{equation}
Restricting $K_\g$ to a simple summand reproduces that summand's intrinsic Killing form.

A simple algebra carries a one-dimensional space of invariant symmetric bilinear forms. The restriction $\langle\cdot,\cdot\rangle|_{\s_l}$ and the Killing form $K_{\s_l}$ are two such forms, so $K_{\s_l}=-N_l\,\langle\cdot,\cdot\rangle|_{\s_l}$ for a single scalar. Negative definiteness on the compact algebra fixes $N_l>0$. The restricted Killing matrix in the orthonormal basis is then $-N_l I_{d_l}$, with eigenvalues the constant $-N_l$. Cross terms between distinct summands vanish as well. For $B_i\in\s_l$ and $B_j\in\s_{l'}$ with $l'\ne l$, the image of $\ad_{B_j}$ lies in $\s_{l'}$, where $\ad_{B_i}$ acts as zero, so $K_\g(B_i,B_j)=\mathrm{tr}_{\g}\!\bigl(\ad_{B_i}\,\ad_{B_j}\bigr)=0$ and $K|_{\s}$ is block-diagonal. Hence $\kappa(K|_{\s})=\max_l N_l/\min_l N_l$, and \eqref{eq:kappa_three_halves} yields the stated ratio.
\end{proof}

\begin{remark}[Frame boundary and blockwise conformal scaling]\label{rmk:orthogonality_necessary}
The theorem covers real invariant-orthonormal frames. An arbitrary change of basis need not preserve total antisymmetry of the raised constants or Identity~\eqref{eq:CCT_eq_minusK}; in a nonorthogonal frame the full inverse Gram matrix enters the raised index. There is a useful exception. Starting from an invariant-orthonormal basis $\{E_{l,a}\}$, set $B_{l,a}=a_lE_{l,a}$ with one nonzero real constant $a_l$ on each simple ideal. Both $CC^{T}|_{\s_l}$ and $-K|_{\s_l}$ then acquire the factor $a_l^2$, so the identity survives. Changes confined to the center are also invisible because $C$ and $K$ vanish there. Generic shears within a nonabelian ideal can break the identity, with no universal lower bound on the residual.
\end{remark}

The scalar $N_l$ in Corollary~\ref{cor:invariant_block_ratio} is determined by the type of $\s_l$, through its dual Coxeter number, and by the normalization isolated in Remark~\ref{rmk:orthogonality_necessary}. The identification holds for every compact simple type and determines the dimensional scaling of the obstruction.

\begin{proposition}[Casimir dictionary for the block scalar]\label{prop:casimir_dictionary}
Let $\s_l$ be a compact simple ideal of Cartan type $X_{r}$, and let $h^{\vee}_l$ be its dual Coxeter number. Fix the positive-definite invariant form $(\cdot,\cdot)_l$ on $\s_l$ by
\[
K_{\s_l}=-2h^{\vee}_l\,(\cdot,\cdot)_l,
\]
the basic normalization in which the highest long root has squared length two and the adjoint quadratic Casimir equals $2h^{\vee}_l$~\cite[Sec.~9.2, Eqs.~(9.16)--(9.17)]{Ginsparg1988}. Then the block scalar of Corollary~\ref{cor:invariant_block_ratio} is
\begin{equation}\label{eq:N_l_casimir}
N_l=2h^{\vee}_l\,\nu_l,\qquad \nu_l:=(B_k,B_k)_l ,
\end{equation}
where $\nu_l$ is the squared basic length shared by the orthonormal generators of the block. The first factor is fixed by the type. The second records the relative normalization between the chosen invariant metric and the basic form, and it equals one exactly for a basic-orthonormal frame.
\end{proposition}

\begin{proof}
By the one-dimensional space of invariant forms used for Corollary~\ref{cor:invariant_block_ratio}, the chosen inner product $\langle\cdot,\cdot\rangle|_{\s_l}$, the Killing form $K_{\s_l}$, and the basic form $(\cdot,\cdot)_l$ are pairwise proportional. That corollary records the first proportionality as $K_{\s_l}=-N_l\,\langle\cdot,\cdot\rangle|_{\s_l}$; the displayed normalization records the second as $K_{\s_l}=-2h^{\vee}_l\,(\cdot,\cdot)_l$. Evaluate both on an orthonormal generator $B_k$ of the block. The Killing values coincide, so $N_l\langle B_k,B_k\rangle=2h^{\vee}_l(B_k,B_k)_l$, and orthonormality $\langle B_k,B_k\rangle=1$ gives \eqref{eq:N_l_casimir}. Positivity of $(\cdot,\cdot)_l$ on the compact form forces $\nu_l>0$, and proportionality of the two forms makes $(B_k,B_k)_l$ independent of which $k\in\s_l$ is taken.
\end{proof}

The dual Coxeter number is tabulated for every simple type. Table~\ref{tab:dual_coxeter} records the values in the same highest-long-root normalization used above~\cite[Eq.~(9.17)]{Ginsparg1988}. The dimension $d_l=\dim\s_l$ is included so that Equation~\eqref{eq:N_l_casimir} becomes explicit once a frame is chosen.

\begin{table}[ht]
\centering
\begin{tabular}{llcc}
\toprule
Type & Compact algebra & $d_l=\dim\s_l$ & $h^{\vee}_l$\\
\midrule
$A_r\ (r\ge 1)$ & $\su(r+1)$ & $r(r+2)$ & $r+1$\\
$B_r\ (r\ge 2)$ & $\mathfrak{so}(2r+1)$ & $r(2r+1)$ & $2r-1$\\
$C_r\ (r\ge 3)$ & $\mathfrak{sp}(2r)$ & $r(2r+1)$ & $r+1$\\
$D_r\ (r\ge 4)$ & $\mathfrak{so}(2r)$ & $r(2r-1)$ & $2r-2$\\
$G_2$ & & $14$ & $4$\\
$F_4$ & & $52$ & $9$\\
$E_6$ & & $78$ & $12$\\
$E_7$ & & $133$ & $18$\\
$E_8$ & & $248$ & $30$\\
\bottomrule
\end{tabular}
\caption{Dual Coxeter numbers and dimensions of the compact simple types, with the rank ranges that keep the families distinct. The basic-normalization block scalar is $N_l=2h^{\vee}_l$, and a general orthonormal frame scales it by $\nu_l$ through \eqref{eq:N_l_casimir}.}
\label{tab:dual_coxeter}
\end{table}

\begin{remark}[Type-dependent dimensional scaling]\label{rmk:type_dependent_scaling}
Equation~\eqref{eq:N_l_casimir} splits the conditioning into two independent factors. The dual Coxeter number is rigid and depends on the type alone; $\nu_l$ depends on how the frame presents the generators. When one uniform frame fixes a common $\nu_l$ across the blocks, that factor cancels from the ratio, and Corollary~\ref{cor:invariant_block_ratio} reads
\[
\kappa(A_{\mathrm{gen}})=\Bigl(\frac{\max_l h^{\vee}_l}{\min_l h^{\vee}_l}\Bigr)^{3/2}.
\]
Inside one classical family the dual Coxeter number grows as a square root of the dimension: $h^{\vee}=\sqrt{d+1}$ for $A_r$, while $B_r,C_r,D_r$ give $h^{\vee}=\Theta(\sqrt{d})$ with respective constants $\sqrt 2,\,1/\sqrt 2,\,\sqrt 2$. A frame whose generators have bounded basic length, $\nu_l=O(1)$, leaves $N_l=\Theta(\sqrt{d_l})$, so $\kappa(A_{\mathrm{gen}})=\Theta\bigl((d_{\max}/d_{\min})^{3/4}\bigr)$ within one family. The generalized Gell--Mann frame of Theorem~\ref{thm:gell_mann_parameters}, normalized by $\Tr(\lambda_a\lambda_b)=2$, has this property. If a separately specified frame satisfies $\nu_l=\Theta(\sqrt{d_l})$, then $N_l=\Theta(d_l)$ and the ratio becomes $\Theta((d_{\max}/d_{\min})^{3/2})$. Full rank and unit operator norm alone do not imply that scaling. The type and the chosen metric both enter through $N_l$, equivalently through the pair $(h^{\vee}_l,\nu_l)$.
\end{remark}

Table~\ref{tab:model_parameters} collects algebraic normalization examples for Equation~\eqref{eq:N_l_casimir}, with the metric and algebraic setting named in each row. In the defining trace metric, the index values in the highest-long-root normalization give $\nu=1/2$ for $\mathfrak{so}(m)$ and $\nu=1$ for $\mathfrak{usp}(2r)$~\cite[Eq.~(9.17)]{Ginsparg1988}. The blockwise Gell--Mann metric on $\su(M)$ has $\nu=2$.

\begin{table}[ht]
\centering
\resizebox{\textwidth}{!}{%
\begin{tabular}{llccccc}
\toprule
Algebraic input & Metric/frame & Type & $h^{\vee}$ & $\nu$ & $N_l$ & $\kappa(A_{\mathrm{gen}})$\\
\midrule
$\su(M)$, $M\ge2$ & blockwise $-\Tr_M(XY)/2$ & $A_{M-1}$ & $M$ & $2$ & $4M$ & $1$\\
$\mathfrak{so}(m)$, $m\ge5$ & defining $-\Tr_m(XY)$ & $B/D$ & $m-2$ & $1/2$ & $m-2$ & $1$\\
$\mathfrak{usp}(2r)$, $r\ge2$ & defining $-\Tr_{2r}(XY)$ & $C_r$ & $r+1$ & $1$ & $2(r+1)$ & $1$\\
$\mathfrak{so}(7)\oplus\mathfrak{usp}(6)$ & common defining trace & $B_3\oplus C_3$ & $5,4$ & $1/2,1$ & $5,8$ & $(8/5)^{3/2}$\\
$\mathfrak{so}(8)\oplus\mathfrak{usp}(6)$ & common defining trace & $D_4\oplus C_3$ & $6,4$ & $1/2,1$ & $6,8$ & $(4/3)^{3/2}$\\
full collective-$SU(2)$ centralizer, $n\ge3$ & blockwise multiplicity trace & $\bigoplus_{J\in\mathcal I_n}A_{m_J-1}$ & $m_J$ & $2$ & $4m_J$ & $(m_+/m_-^+)^{3/2}$\\
\bottomrule
\end{tabular}}
\caption{Algebraic block scalars in named real invariant metrics. The cross-type rows use abstract direct sums. The final row is the full collective-$SU(2)$ centralizer under the blockwise metric of Theorem~\ref{thm:su2_symmetric_obstruction}; its traceless center has dimension $q-1$. Physical generating sets, typed access constructions, and end-to-end comparisons are separate inputs.}
\label{tab:model_parameters}
\end{table}

The two frame instances quoted in Section~\ref{sec:condition_number_barrier} follow from the same dictionary. The bounded-normalization corollary lowers the exponent to three quarters.

\begin{corollary}[Blockwise Gell--Mann metric, type-$A$ family]\label{cor:obstruction_hs}
Let $\s=\bigoplus_l\su(M_l)$ be semisimple of type $A$, each block carried in the generalized Gell-Mann frame normalized by $\mathrm{tr}(\lambda_a\lambda_b)=2$. Then $\nu_l=2$ and $N_l=2h^{\vee}_l\nu_l=4M_l$, so
\[
\kappa(A_{\mathrm{gen}})=\Bigl(\frac{M_{\max}}{M_{\min}}\Bigr)^{3/2}=\Bigl(\frac{d_{\max}+1}{d_{\min}+1}\Bigr)^{3/4}=\Theta\!\bigl((d_{\max}/d_{\min})^{3/4}\bigr).
\]
The inner product is fixed separately on each ideal by $\langle X,Y\rangle_l=-\mathrm{tr}_{M_l}(XY)/2$. It is not the ambient Hilbert-space trace of an isotypic embedding.
\end{corollary}

\begin{proof}
Type $A_{M_l-1}$ has $h^{\vee}_l=M_l$ (Table~\ref{tab:dual_coxeter}), and the Gell-Mann normalization fixes the squared basic-length at the bounded value $\nu_l=2$ of Remark~\ref{rmk:type_dependent_scaling}. Proposition~\ref{prop:casimir_dictionary} gives $N_l=4M_l$, and $M_l=\sqrt{d_l+1}$ turns the ratio of Theorem~\ref{thm:obstruction_general} into the three-quarters power.
\end{proof}

The proof of the $SU(2)$-symmetric obstruction stated in Section~\ref{sec:condition_number_barrier} runs through the same dictionary.

\begin{proof}[Proof of Lemma~\ref{lem:collective_su2_centralizer} and Theorem~\ref{thm:su2_symmetric_obstruction}]
Schur's lemma applied to the displayed isotypic decomposition gives $\mathfrak c_{\uu}=\bigoplus_J I_{V_J}\otimes\uu(m_J)$. The trace of the $J$-th block is $r_J\Tr(X_J)$, which gives the stated hyperplane inside $\su(2^n)$. Commutators remove the scalar part of every matrix block, so the derived algebra is the direct sum of the ideals $\su(m_J)$ with $m_J\ge2$. A central element has the form $\bigoplus_J i\theta_J I_{V_J\otimes M_J}$ and is traceless exactly when $\sum_Jr_Jm_J\theta_J=0$. This is one nonzero linear constraint on the $q$ sector scalars.

The multiplicity formula follows from weight counting. The weight-$J$ subspace of $(\C^2)^{\otimes n}$ has dimension $\binom n{n/2-J}$, and every spin-$L$ representation with $L\ge J$ contributes one vector to that weight. Subtracting the weight-$(J+1)$ dimension leaves the number of spin-$J$ copies.

For the extrema, rewrite the multiplicity as
\[
m_J=\frac{2J+1}{n/2+J+1}\binom n{n/2-J}.
\]
The ratio of consecutive admissible terms is
\[
\frac{m_{J+1}}{m_J}
=\frac{(n/2-J)(2J+3)}{(2J+1)(n/2+J+2)}.
\]
It is at least one exactly when $4(J+1)^2\le n+2$. Thus the maximum lies at $J=\sqrt n/2+O(1)$, and the local central-binomial estimate gives
\[
m_+\sim \sqrt{\frac{8}{\pi e}}\,\frac{2^n}{n}.
\]
The last active sector has $J=n/2-1$ and multiplicity $n-1$. Unimodality, together with the opposite endpoint value, gives $m_-^+=n-1$ for $n\ge5$; direct substitution gives $m_-^+=2$ for $n=3,4$. For even $n$, the singlet formula at $J=0$ is the Catalan number displayed in Remark~\ref{rmk:simple_summand_obstruction}, whose Stirling expansion is $\sqrt{8/\pi}\,2^n/n^{3/2}$.

Under the ambient metric, the relevant ratio is
\[
\frac{m_J}{r_J}=\frac{1}{n/2+J+1}\binom n{n/2-J}.
\]
Its consecutive-term ratio is $(n/2-J)/(n/2+J+2)<1$. The maximum is therefore at the smallest spin and is asymptotic to $\sqrt{8/\pi}\,2^n/n^{3/2}$, while the last active sector has $m_J/r_J=1$. Raising these block ratios to the three-halves power proves both displayed conditioning laws.

Each nontrivial matrix block is of type $A_{m_J-1}$, with dual Coxeter number $m_J$. Proposition~\ref{prop:casimir_dictionary} gives $N_J=2m_J\nu_J$. Corollary~\ref{cor:invariant_block_ratio}, with extrema restricted to $\mathcal I_n$, gives Equation~\eqref{eq:kappa_su2}.
\end{proof}

\section{Classical benchmark proof}\label{app:classical_benchmark}

\begin{proof}[Proof of Theorem~\ref{thm:classical_benchmark_main}]
The count uses exact field operations first and accounts for coefficient bit length
afterwards.

\textbf{Step 1: indexed sparse input.}
CSR and CSC indices for the nonzero structure constants take
$O(\operatorname{nnz}(f))=O(ds_{adj})$ insertions and storage.

\textbf{Step 2: calculation of the Killing form.}
For the factorization $K=A^TB$,
\begin{equation}
\operatorname{flops}(A^TB)
=\sum_{j=1}^{d}\ \sum_{p\in\operatorname{supp}(B_{\cdot j})}
\operatorname{nnz}\bigl((A^T)_{\cdot p}\bigr)
\leq d\,s_{adj}s_{int}.
\label{eq:spmm_flop_count}
\end{equation}
The sparse-accumulator model also charges input indexing and dense-output
initialization or writes, giving
$O(ds_{adj}s_{int}+ds_{adj}+d^2)$ time and $O(d)$ accumulator memory in addition
to the stored matrices~\cite{Buluc2008}. The $d^2$ term is absorbed by the later
rank step in either case.

\textbf{Step 3a: general case.}
The commutator matrix has shape $d\times d^2$. Ordinary elimination has at most
$d$ pivots. Each pivot can update $O(d)$ rows across $O(d^2)$ columns, so the
worst-case count is
\begin{equation}
\sum_{k=0}^{d-1}O\bigl((d-k)(d^2-k)\bigr)=O(d^4).
\label{eq:rectangular_rank_count}
\end{equation}
This step returns a basis matrix $M$ for $[\g,\g]$. Forming $M^TK$ and computing
its kernel cost $O(d^3)$ more operations.

\textbf{Step 3b: compact case.}
Compactness gives $\rfrak=\ker K$, so dense elimination acts only on a
$d\times d$ matrix and costs $O(d^3)$ field operations.

Combining the steps proves Equation~\eqref{eq:classical_two_route_bounds}.

For the exact bit model, multiply all supplied dyadic entries by their common
$b$-bit denominator. A rank-$d$ fraction-free elimination stores minors of an
integer matrix. Hadamard's inequality bounds their bit length by
$O(d(b+\log d))$. Multiplying this word length by the $O(d^4)$ or $O(d^3)$ field
operation count, using fast integer arithmetic, gives the two bounds in
Equation~\eqref{eq:classical_bit_bounds}. Complex dyadics are handled by pairs of
integer components. This argument does not cover floating-point numerical rank.
\end{proof}

\section{Gell-Mann basis analysis}\label{app:gell_mann}

We analyze the actual real compact Gell--Mann frame for $\su(N)$, with $d=N^2-1$.

\begin{theorem}[Parameters in the Gell-Mann basis]\label{thm:gell_mann_parameters}
Let $N\ge2$ be an integer, let $\lambda_a$ be the standard Hermitian generalized Gell--Mann matrices with $\Tr(\lambda_a\lambda_b)=2\delta_{ab}$, and set $B_a=i\lambda_a$. Regard their span as a real Lie algebra with $\langle X,Y\rangle=-\Tr(XY)/2$. Define
\[
D_N=\begin{cases}
N^2/2+N-2,&N\text{ even},\\
(N^2+2N-3)/2,&N\text{ odd},
\end{cases}
\qquad S_N=\max\{6N-10,D_N\}.
\]
Then the exact algebraic parameters are
\[
F=2,\qquad s_{adj}=s_{\mathrm{out}}=S_N,\qquad
s_{int}=s_{\mathrm{br}}=N-1,
\]
and
\[
N_k=4N,\qquad \beta=1,\qquad K=-4N I_d,
\qquad \kappa(A_{\mathrm{gen}})=1.
\]
For the conservative sparse normalizations of Section~\ref{sec:workflow},
\begin{align}
\rho_{\mathrm{BE}}(K)
&=\frac{S_N(N-1)}{N}=\Theta(d),\nonumber\\
\rho_{\mathrm{BE}}(A_{\mathrm{gen}})
&=\left(\frac{S_N(N-1)}{N}\right)^{3/2}=\Theta(d^{3/2}).
\label{eq:gell_mann_rho}
\end{align}
\end{theorem}

\begin{unnumberedremark}[Access-model scope]
These are algebraic normalization bounds. They do not supply reversible
support enumerators for this basis.
\end{unnumberedremark}
\begin{proof}
Write the Hermitian frame as the symmetric matrices $S_{ab}=E_{ab}+E_{ba}$, the antisymmetric matrices $A_{ab}=-iE_{ab}+iE_{ba}$ for $a<b$, and the diagonal matrices
\[
H_l=\sqrt{\frac{2}{l(l+1)}}\,\operatorname{diag}(\underbrace{1,\ldots,1}_{l},-l,0,\ldots,0),
\qquad 1\le l<N.
\]
Multiplication of matrix units gives real structure constants for $B=i\lambda$. The largest coefficient has magnitude two and is attained in every embedded $\su(2)$ triple, hence $F=2$.

For an off-diagonal pair $(a,b)$, let $t_{ab}$ be the number of diagonal generators $H_l$ whose $a$- and $b$-entries differ. Brackets with off-diagonal generators sharing one endpoint contribute $4(N-2)$ nonzero $(j,k)$ positions. The companion $S_{ab}$--$A_{ab}$ bracket has $t_{ab}$ diagonal outputs, and the $t_{ab}$ diagonal inputs each have one companion output. The row support is therefore
\[
4(N-2)+2t_{ab}\le 6N-10,
\]
with equality because $t_{1N}=N-1$. For $H_l$, the three equal-entry groups have sizes $l$, $1$, and $N-l-1$, so the row support is
\[
2\left[\binom N2-\binom l2-\binom{N-1-l}{2}\right].
\]
Its maximum over $l$ is $D_N$, which proves $s_{adj}=S_N$. The interaction count has a separate exhaustive split. Disjoint off-diagonal pairs commute; pairs sharing one endpoint fix one off-diagonal output; and a diagonal--off-diagonal pair fixes one companion output. Only the companion pair $S_{ab},A_{ab}$ can have several diagonal partners, and their number is $t_{ab}\le N-1$. The pair $S_{1N},A_{1N}$ attains all $N-1$ diagonal directions, so $s_{int}=N-1$. Total antisymmetry gives $s_{\mathrm{out}}=s_{adj}$ and $s_{\mathrm{br}}=s_{int}$.

The defining-representation trace identity gives $K(X,Y)=2N\Tr(XY)=-4N\langle X,Y\rangle$. Thus $K=-4NI_d$, and Lemma~\ref{lem:invariant_reduction} gives $CC^T=4NI_d$ and $A_{\mathrm{gen}}^TA_{\mathrm{gen}}=(4N)^3I_d$. Hence every $N_k$ equals $4N$, so $\beta=1$ and $\kappa(A_{\mathrm{gen}})=1$.

Finally, the sparse-access formulas give $\alpha_{\mathrm{cmp}}=4S_N(N-1)$ and $\alpha_{\mathrm{gen}}=8[S_N(N-1)]^{3/2}$. Dividing by $\sigma_{\max}(K)=4N$ and $\sigma_{\max}(A_{\mathrm{gen}})=8N^{3/2}$ proves Equation~\eqref{eq:gell_mann_rho}. Since $S_N=\Theta(N^2)$ and $d=N^2-1$, the two asymptotic bounds follow.
\end{proof}

The blockwise trace normalization used here fixes the basic length $\nu_l=2$ in Corollary~\ref{cor:obstruction_hs}. The single simple block remains spectrally conditioned, but its conservative block encoding has the slack shown in Equation~\eqref{eq:gell_mann_rho}.

\printbibliography

\end{document}